\documentclass[a4paper,10pt]{article}
\usepackage[utf8]{inputenc}
 \usepackage{amsmath,amsfonts,amssymb}
 \usepackage{amsthm}
 \usepackage{mathtools}
 \usepackage{amsopn}
 \usepackage{tabularx,lipsum,environ}
 \usepackage{paralist}
 \usepackage{microtype}
 \usepackage{authblk}
 \usepackage{hyperref}
 \usepackage{fancyhdr}
 \usepackage{rotating}
 \usepackage{overpic}
 \usepackage{ucs}
 \usepackage{enumerate}
 \usepackage{graphicx}
 \usepackage{booktabs}
 \usepackage{varwidth}
 \usepackage{subfig}

\usepackage[color=green!20]{todonotes}

\usepackage{algorithm}
\usepackage[noend]{algpseudocode}

\usepackage{color}
\definecolor{darkred}{rgb}{0.8,0,0}

\newtheorem{observation}{Observation}
\newtheorem{myclaim}{Claim}

\newcommand{\figref}[1]{Figure~\ref{#1}}

\usepackage{comment}
\usepackage{tikz}
\tikzstyle{vertex}=[circle, draw, inner sep=1pt, minimum width=6pt]
\usetikzlibrary{decorations,decorations.pathmorphing,decorations.pathreplacing,fit,arrows,calc}

\newcommand{\vc}{\mathsf{vc}}
\newcommand{\nd}{\mathsf{nd}}
\newcommand{\mw}{\mathsf{mw}}
\newcommand{\cw}{\mathsf{cw}}
\newcommand{\tw}{\mathsf{tw}}
\newcommand{\td}{\mathsf{td}}
\newcommand{\pw}{\mathsf{pw}}

\newcommand{\safe}{\mathsf{s}}
\newcommand{\csafe}{\mathsf{cs}}

\newcommand{\vi}{\mathsf{vi}}

\renewcommand{\mid}{:}

\newcommand{\DP}{\mathtt{dp}} 
\newcommand{\num}{\mathtt{n}}             
\newcommand{\bnum}{\overline{\mathtt{n}}} 
\newcommand{\smllst}{\mathtt{s}} 
\newcommand{\lrgst}{\mathtt{l}}  
\newcommand{\mindiff}{\mathtt{d}}  
\newcommand{\asmllst}{\mathtt{s}^{\mathtt{a}}} 
\newcommand{\alrgst}{\mathtt{l}^{\mathtt{a}}}  

\newcommand{\true}{\mathsf{true}}
\newcommand{\false}{\mathsf{false}}

\usepackage{tabularx,environ}
\usepackage{mathtools}

\renewcommand\leq\leqslant
\renewcommand\geq\geqslant
\renewcommand\le\leqslant
\renewcommand\ge\geqslant

\makeatletter
\newcommand{\problemtitle}[1]{\gdef\@problemtitle{#1}}
\newcommand{\probleminput}[1]{\gdef\@probleminput{#1}}
\newcommand{\problemquestion}[1]{\gdef\@problemquestion{#1}}
\NewEnviron{myproblem}{
  \problemtitle{}\probleminput{}\problemquestion{}
  \BODY
  \par\addvspace{.5\baselineskip}
  \noindent \normalsize
  \begin{tabularx}{\textwidth}{@{\hspace{\parindent}} l X c}
    \multicolumn{2}{@{\hspace{\parindent}}l}{\normalsize\@problemtitle} \\
    \normalsize \textbf{Input:} & \normalsize \@probleminput \\
    \normalsize \textbf{Question:} & \normalsize \@problemquestion
  \end{tabularx}
  \par\addvspace{.5\baselineskip}
}
\makeatother

\title{Parameterized Complexity of Safe Set\thanks{Partially supported by JSPS and MAEDI under the Japan-France Integrated Action Program (SAKURA) Project GRAPA 38593YJ,
and by JSPS/MEXT KAKENHI Grant Numbers JP24106004, JP17H01698, JP18K11157, JP18K11168, JP18K11169, JP18H04091, 18H06469.}}

\author[1]{R\'{e}my Belmonte}
\author[2]{Tesshu Hanaka}
\author[3]{Ioannis Katsikarelis}
\author[3]{Michael Lampis}
\author[4]{Hirotaka Ono}
\author[5]{Yota Otachi}

\affil[1]{The University of Electro-Communications, Chofu, Tokyo, 182-8585, Japan}
\affil[2]{Chuo University, Bunkyo-ku, Tokyo, 112-8551, Japan}
\affil[3]{Universit\'{e} Paris-Dauphine, PSL University, CNRS, LAMSADE, Paris, France}
\affil[4]{Nagoya University, Nagoya, 464-8601, Japan}
\affil[5]{Kumamoto University, Kumamoto, 860-8555, Japan}

\theoremstyle{plain}
\newtheorem{theorem}{Theorem}
\newtheorem{lemma}[theorem]{Lemma}
\newtheorem{corollary}[theorem]{Corollary}

\newtheorem{proposition}[theorem]{Proposition}

\newlength{\defbaselineskip}
\newcommand{\setlinespacing}[2]%
           {\setlength{\baselineskip}{#1 \defbaselineskip}}

\newlength{\btw}
\newlength{\stw}
\newsavebox\tmpbox 

\date{}

\begin{document}

\maketitle

\begin{abstract}
In this paper we study the problem of finding a small safe set $S$ in a graph $G$,
i.e.\ a non-empty set of vertices such that no connected component of $G[S]$
is adjacent to a larger component in $G - S$.
We enhance our understanding of the problem from the viewpoint of parameterized complexity by showing that
(1) the problem is W[2]-hard when parameterized by the pathwidth $\pw$ and cannot be solved in time $n^{o(\pw)}$ unless the ETH is false,
(2) it admits no polynomial kernel parameterized by the vertex cover number $\vc$ unless $\mathrm{PH} = \Sigma^{\mathrm{p}}_{3}$, but
(3) it is fixed-parameter tractable (FPT) when parameterized by the neighborhood diversity $\nd$, and
(4) it can be solved in time $n^{f(\cw)}$ for some double exponential function $f$ where $\cw$ is the clique-width.
We also present (5) a faster FPT algorithm when parameterized by solution size.
\end{abstract}


\section{Introduction}

Let $G = (V,E)$ be a graph.
For a vertex set $S \subseteq V(G)$, we denote by $G[S]$ the subgraph of $G$ induced by $S$,
and by $G - S$ the subgraph induced by $V \setminus S$.
If $G[S]$ is connected, we also say that $S$ is connected.
A vertex set $C \subseteq V$ is a \emph{component} of $G$ if $C$ is an inclusion-wise maximal connected set.
Two vertex sets $A, B \subseteq V$ are \emph{adjacent}
if there is an edge $\{a,b\} \in E$ such that $a \in A$ and $b \in B$.
Now, a non-empty vertex set $S \subseteq V$ of a graph $G = (V,E)$ is a \emph{safe set}
if no connected component $C$ of $G[S]$
has an adjacent connected component $D$ of $G -S$ with $|C| < |D|$.
A safe set $S$ of $G$ is a \emph{connected safe set} if $G[S]$ is connected.

The \emph{safe number} $\safe(G)$ of $G$ is the size of a minimum safe set of $G$,
and the \emph{connected safe number} $\csafe(G)$ of $G$ is the size of a minimum connected safe set of $G$.
It is known~\cite{FMS16} that $\safe(G) \le \csafe(G) \le 2 \cdot \safe(G) - 1$.

Fujita, MacGillivray, and Sakuma~\cite{FMS16} introduced the concept of safe sets
motivated by a facility location problem that can be used to design a safe evacuation plan.
Subsequently, Bapat et al.~\cite{BapatFLMMST_weighted} observed that a safe set can control the consensus of the underlying network
with a majority of each part of the subnetwork induced by a component in the safe set and an adjacent component in the complement, thus the minimum size of a safe set can be used as a vulnerability measure of the network.
That is, contrary to its name, having a small safe set could be unsafe for a network.
Both the combinatorial and algorithmic aspects of the safe set problem have already been extensively studied~\cite{AguedaCFLMMMNOS_structural,FujitaF_preprint,EhardR17arxiv_approx,FujitaJPSarxiv_cycle}.

In this paper, we study the problem of finding a small safe set mainly from the parameterized-complexity point of view.
We show a number of both tractability and intractability results,
that highlight the difference in complexity that this parameter exhibits compared to similarly defined vulnerability parameters.


\subsection*{Our results}
Our main results are the following (see also \figref{fig:width-parameters}).
\begin{enumerate}
  \item Both problems are W[2]-hard parameterized by the pathwidth $\pw$
  and cannot be solved in time $n^{o(\pw)}$ unless the Exponential Time Hypothesis (ETH) fails,
  where $n$ is the number of vertices.
  
  \item They do not admit kernels of polynomial size when parameterized by vertex cover number
  even for connected graphs unless $\mathrm{PH} = \Sigma^{\mathrm{p}}_{3}$.

  \item Both problems are fixed-parameter tractable (FPT) when parameterized by neighborhood diversity.

  \item Both problems can be solved in XP-time when parameterized by clique-width.

  \item Both problems can be solved in $O^*(k^{O(k)})$ time\footnote{%
    The $O^*(\cdot)$ notation omits the polynomial dependency on the input size.}
  when parameterized by the solution size $k$.
\end{enumerate}

The W[2]-hardness parameterized by pathwidth complements the known FPT result when parameterized
by the solution size~\cite{AguedaCFLMMMNOS_structural}, since for every graph the size of the solution is at least half of the graph's pathwidth (see Section~\ref{sec:prel}).
The $n^{o(\pw)}$-time lower bound is tight since there is an $n^{O(\tw)}$-time algorithm~\cite{AguedaCFLMMMNOS_structural},
where $\tw$ is the treewidth.
The second result also implies that there is no polynomial kernel parameterized by solution size,
as the vertex cover number is an upper bound on the size of the solution.
The third result marks the first FPT algorithm by a parameter that is incomparable to the solution size.
The fourth result implies XP-time solvability for all the parameters mentioned in this paper
and extends the result for treewidth from~\cite{AguedaCFLMMMNOS_structural}.
The fifth result improves the known algorithm~\cite{AguedaCFLMMMNOS_structural} that uses Courcelle's theorem.

\begin{figure}[bth]
  \centering
  \includegraphics[scale=0.8]{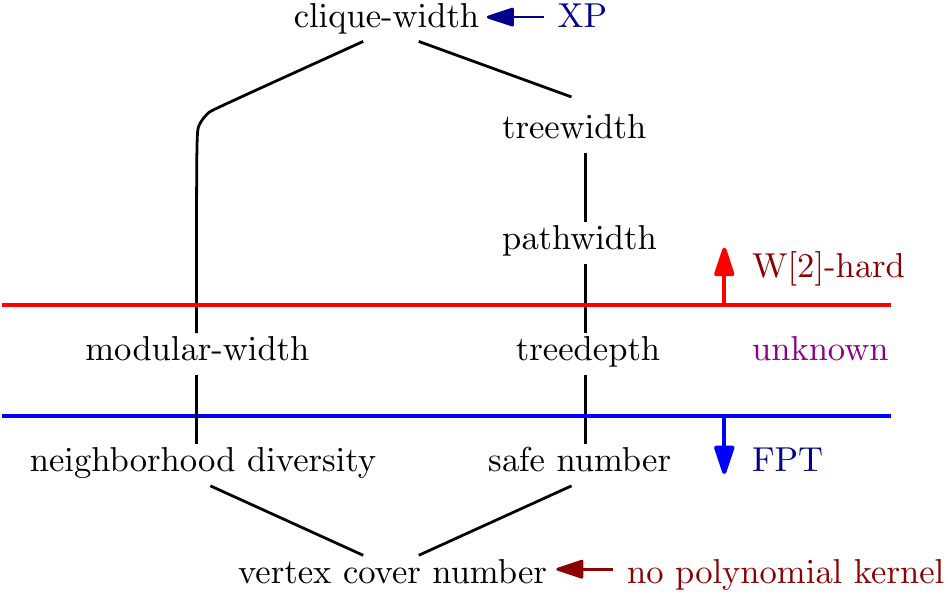}
  \caption{Graph parameters and an overview of the parameterized complexity landscape for SS and CSS.
    Connections between two parameters imply the existence of a function in the one above (being in this sense more general) that lower bounds the one below.}
  \label{fig:width-parameters}
\end{figure}


\subsection*{Previous work}
In the first paper on this topic, Fujita et al.~\cite{FMS16} showed that the problems are NP-complete in general.
Their hardness proof implies that the parameters are hard to approximate within a factor of 1.3606
(\cite{AguedaCFLMMMNOS_structural}).
They also showed that a minimum connected safe set in a tree can be found in linear time.

Bapat et al.~\cite{BapatFLMMST_weighted} considered the problems on vertex-weighted graphs,
where the problems are naturally generalized.
They showed that these are weakly NP-complete even for weighted stars
(thus for all parameters generalizing vertex cover number and graph classes like interval and split graphs).
On the other hand, they showed that the problems can be solved in $O(n^{3})$ time for weighted paths.
%
Ehard and Rautenbach~\cite{EhardR17arxiv_approx} presented a PTAS for the connected safe number of a weighted tree.
Fujita et al.~\cite{FujitaJPSarxiv_cycle} showed  among other results
that the problems can be solved in linear time for weighted cycles.

\'{A}gueda et al.~\cite{AguedaCFLMMMNOS_structural} studied their unweighted versions.
They presented an XP algorithm for graphs of bounded treewidth and
showed that the problems can be solved in polynomial time for interval graphs,
while they are NP-complete for split and bipartite planar graphs of maximum degree at most~7.
Observing that the treewidth of a graph is bounded by a function of its safe number,
they also showed that the problems are FPT parameterized by solution size.


\section{Definitions and Preliminaries}\label{sec:prel}

The problems studied in this paper are formalized as follows:
 \begin{myproblem}
   \problemtitle{\textsc{Safe Set} (SS)}
   \probleminput{A graph $G = (V,E)$ and an integer $k$.}
   \problemquestion{Is there a safe set $S \subseteq V$ of size at most $k$?}
 \end{myproblem}
 \begin{myproblem}
   \problemtitle{\textsc{Connected Safe Set} (CSS)}
   \probleminput{A graph $G = (V,E)$ and an integer $k$.}
   \problemquestion{Is there a connected safe set $S \subseteq V$ of size at most $k$?}
 \end{myproblem}
 \noindent In our positive results, we also study the optimization versions of both problems.
 
A set $K \subseteq V$ is a \emph{dominating set} of $G = (V,E)$
 if each vertex in $G$ is either included in $K$ or has a neighbor in $K$;
 that is, $N[v] \cap K \ne \emptyset$.
 \textsc{Dominating Set} asks to for a dominating set of size at most $k$
 and is known to be W[2]-complete when parameterized by the solution size $k$~\cite{DowneyF95}. A related problem is the following:
 
 \begin{myproblem}
   \problemtitle{\textsc{Red-Blue Dominating Set} (RBDS)}
   \probleminput{A bigraph $G = (R,B;E)$ and an integer $k$.}
   \problemquestion{Is there $D \subseteq B$ of size at most $k$ that dominates all vertices in $R$?}
 \end{myproblem}

We assume that the reader is familiar with the concepts relating to fixed-parameter tractability.
See~\cite{CyganFKLMPPS15} (and references therein) for the definitions of relevant notions in parameterized complexity theory.
We let $N[v]$ denote the \emph{closed neighborhood} of vertex $v$, i.e.\ the set containing $v$ and all vertices adjacent to it.
Furthermore, for a positive integer $k$, we denote the set $\{1,\dots,k\}$ by $[k]$.


 \subsection*{Graph parameters}\label{ssec:graph-parameters}
We recall relationships among some graph parameters used in this paper,
and give a map of the parameters with the results.
The graph parameters we explicitly use in this paper are 
vertex cover number $\vc$, pathwidth $\pw$, neighborhood diversity $\nd$, and clique-width $\cw$.
They are situated in the hierarchy of well-studied graph parameters, along with safe number $\safe$,
as depicted in \figref{fig:width-parameters}.

The treewidth $\tw(G)$, pathwidth $\pw(G)$, treedepth $\td(G)$, and vertex cover number $\vc(G)$ of a graph $G$
can be defined as the minimum of the maximum clique-size ($-1$ for $\tw(G)$, $\pw(G)$, and $\vc(G)$)
among all supergraphs of $G$ that are of type chordal, interval, trivially perfect and threshold, respectively.
This gives us $\tw(G) \le \pw(G) \le \td(G) - 1 \le \vc(G)$ for every graph $G$.
One can easily see that $\safe(G) \le \vc(G)$ and $\td(G) \le 2 \safe(G)$. (See also the discussion in \cite{AguedaCFLMMMNOS_structural}.)
This justifies the hierarchical relationships among them in \figref{fig:width-parameters}.

Although modular-width $\mw(G)$ and neighborhood diversity $\nd(G)$ are
incomparable to most of the parameters mentioned above,
all these parameters are generalized by clique-width, while the vertex cover number is their specialization.

More formally, a \emph{path decomposition} of a graph $G = (V,E)$
 is a sequence $(X_{1}, \dots, X_{r})$ of subsets (called \emph{bags}) of $V$ such that:
 \begin{enumerate}
 \item $\bigcup_{1 \le i \le r} X_{i} = V$, 
 \item for each $e \in E$, there exists $X_{i}$ such that $e \subseteq X_{i}$, and 
 \item each $v \in V$ appears in consecutive sets in the sequence.
 \end{enumerate}
 The \emph{width} of a path decomposition $(X_{1}, \dots, X_{r})$ is $\max_{1 \le i \le r} |X_{i}|$.
 The \emph{pathwidth} $\pw(G)$ of $G$ is the minimum width over all path decompositions of $G$.

 Next, the clique-width of a graph measures the simplicity of the graph
 as the number of labels required to construct the graph~\cite{CourcelleO00}.
 In a vertex-labeled graph, an \emph{$i$-vertex} is a vertex of label $i$ .
 A \emph{$c$-expression} is a rooted binary tree such that:
 \begin{itemize}
   \setlength{\itemsep}{0pt}
   \item each leaf has label $\circ_{i}$ for some $i \in [c]$,
   \item each non-leaf node with two children has label $\cup$, and
   \item each non-leaf node with exactly one child has label $\rho_{i,j}$ or $\eta_{i,j}$ ($i, j \in [c]$, $i \ne j$).
 \end{itemize}
 Each node in a $c$-expression represents a vertex-labeled graph as follows:
 \begin{itemize}
   \setlength{\itemsep}{0pt}
   \item a $\circ_{i}$-node represents a single-vertex graph with one $i$-vertex;
   \item a $\cup$-node represents the disjoint union of the labeled graphs represented by its children;
   \item a $\rho_{i,j}$-node represents the labeled graph obtained from the one represented by its child by replacing
   the labels of all $i$-vertices with $j$;
   \item an $\eta_{i,j}$-node represents the labeled graph obtained from the one represented by its child
   by adding all possible edges between the $i$-vertices and the $j$-vertices.
 \end{itemize}
 
 A $c$-expression gives the graph represented by its root.
 The \emph{clique-width} $\cw(G)$ of $G$
 is the minimum integer $c$ such that some $c$-expression represents a graph isomorphic to $G$.
 It is known that for any constant $c$,
 one can compute a $(2^{c+1}-1)$-expression of a graph of clique-width $c$ in $O(n^{3})$ time~\cite{HlinenyO08,OumS06,Oum08}.
 
 A $c$-expression of a graph is \emph{irredundant} if for each edge $\{u,v\}$,
 there is exactly one node $\eta_{i,j}$ that adds the edge between $u$ and $v$.
 A $c$-expression of a graph can be transformed into an irredundant one with $O(n)$ 
 nodes in linear time~\cite{CourcelleO00}.
 In what follows, we assume that $c$-expressions are irredundant.
 
 
 \subsection*{Some observations}
 Here we list some simple observations that are not directly related to our main results yet may still be of interest.
 \begin{observation}
   Given a graph $G$,
   the safe number and the connected safe number can be approximated in polynomial time
   within a factor of $\safe(G)+1$.
 \end{observation}
  \begin{proof}
  If a graph is not connected, we can apply the following algorithm for each component 
 and output a found set of the minimum size. Hence we assume that $G$ is connected in what follows.
  
  We guess the safe number $\safe(G)$ and denote the hypothetical value by $s \le |V(G)|$.
  We start with an arbitrary connected set $S$ such that $|S| = s+1$.
  Assume that $G - S$ has a component $C$ of size more than $s$.
  Let $C' \subseteq C$ be a connected set of size $s+1$ that has a neighbor in $S$.
  We put all the vertices of $C'$ into $S$.
  We repeat this process until every component of $G-S$ becomes of size at most $s$.
  
  Observe that the resulting set $S$ is a safe set since $G[S]$ is connected and
  $G-S$ has no component of size larger than $|S| \ge s$.
  Observe that each time we added some vertices to $S$ (even at the initialization of $S$), 
  we added a connected set $X$ of size $s+1$.
  Thus, for any safe set $S$ of size $s$, it holds that $S \cap X \ne \emptyset$.
  This means that we add at most $s+1$ vertices to include each vertex in an optimal solution.
  Therefore the output is a connected safe set of size at most $\safe(G) \cdot (\safe(G)+1)$ provided that the guess $s$ is correct.
  \end{proof}

 The following observation implies that when parameterized by the solution size and the maximum degree $\Delta(G)$,
 SS and CSS admit FPT algorithms for general graphs and polynomial kernels for connected graphs.
 Recall that the problems are NP-complete for graphs of bounded maximum degree~\cite{AguedaCFLMMMNOS_structural},
 and do not admit a polynomial kernel even for connected graphs when parameterized only by solution size (Corollary~\ref{cor:no_poly-kernel}).
 \begin{observation}\label{obs:degree}
   For every connected graph $G$, $|V(G)| \le \safe(G) + \safe(G)^{2} \cdot \Delta(G)$.
 \end{observation}
  \begin{proof}
  If $S$ is a safe set, then there are at most $|S| \cdot \Delta(G)$ components in $G-S$
  and each of them has size at most $|S|$.
   
  \end{proof}
 
 A related observation is that both SS and CSS admit kernels of order roughly $k^{O(k)}$.
 Here we need the (already mentioned) fact that $\td(G) \le 2 \safe(G)$.
 The treedepth $\td(G)$ of a graph $G$ is defined recursively as follows:
 \begin{itemize}
   \item if $G$ has one vertex only, then $\td(G) = 1$;
   \item if $G$ has the components $C_{1}, \dots, C_{q}$ with $q \ge 2$, then $\td(G) = \max_{1 \le i \le q} \td(G[C_{i}])$;
   \item otherwise, $\td(G) = 1 + \min_{v \in V(G)} \td(G) - \{v\}$.
 \end{itemize}
 
 \begin{observation}
 \label{obs:safe-vs-td} 
 $\td(G) \le 2 \safe(G)$ for every graph $G$.
 \end{observation}
  \begin{proof}
  Let $S$ be a safe set of $G$ with size $k \coloneqq \safe(G)$.
  Observe that $\td(G) \le k + \td(G - S)$.
  Now since $S$ is a safe set, each component of $G-S$ has size at most $k$.
  A graph of $k$ vertices has treedepth at most $k$,
  and thus $\td(G - S) \le k$.
   
  \end{proof}
 
 \begin{observation}\label{obs:kernel} There exists a polynomial-time algorithm
 which, given a connected instance of SS, or CSS, produces an equivalent
 instance with at most $\safe(G)+(2\safe(G))^{2\safe(G)}$ vertices.
 \end{observation}
 
  \begin{proof}
  
  Let $k:=\safe(G)$. We first observe that if a vertex $u$ has degree at least
  $2k$ then any safe set (connected or not) must contain $u$, because otherwise
  $u$ together with the (at least $k$) of its neighbors which are not in the
  solution form a component of size $k+1$ or more.  Therefore, if $G$ has $k+1$
  or more vertices of degree at least $2k$ we immediately produce a trivial No
  instance. Let $H$ be the set of vertices of degree at least $2k$.
  
  If we have $|H|\le k$ we need to show that $|V\setminus H|\le (2k)^{2k}$ in any
  Yes instance. However, in a Yes instance, 
  $\td(G) \le 2k$ by Observation~\ref{obs:safe-vs-td}. 
 This implies that
  $\td(G - H)\le 2k$. Furthermore, the maximum degree of $G - H$ is at most $2k$. It is now an easy observation that a graph with
  tree-depth $2k$ and maximum degree $2k$ cannot have more than $(2k)^{2k}$
  vertices. So, if $V\setminus H$ is larger than that we reject, otherwise we
  have the promised bound.  
\end{proof}
 
 Observation \ref{obs:kernel} is interesting as a demonstration of the
 algorithmic difficulties posed by the fact that our problems are not closed
 under taking induced subgraphs (unlike most graph parameters). Since we succeed
 in bounding the number of vertices of high degree, one may be tempted to
 believe that Observation \ref{obs:degree} would then allow us to obtain a
 polynomial kernel (as in the remainder of the graph $\Delta=O(k)$). This does
 not work, as a safe set of the original graph does not necessarily remain a
 safe set of the graph induced by low-degree vertices. As a result, we are
 forced to argue indirectly through a more well-behaved parameter (tree-depth).
 The results of Section \ref{sec:kernel} indicate that this is most likely to be
 inevitable, as the problem does not admit a polynomial kernel unless $\mathrm{PH} = \Sigma^{\mathrm{p}}_{3}$.
 
 
 \subsection*{Similar graph parameters}
 While the safe number is a recently-introduced parameter,
 it bears some similarities with other graph parameters that are somewhat older and relatively well-studied.
 These are the vertex integrity and the $\ell$-component order connectivity.
 Here we discuss some known results on these parameters and compare them to those on the safe number.
 
 The \emph{vertex integrity} $\vi(G)$~\cite{BarefootES87} of a graph $G$ is the minimum integer $k$ such that
 there is a vertex set $S \subseteq V$ such that $|S|$ plus the maximum size of the components of $G - S$ is at most $k$.
 Fujita and Furuya~\cite{FujitaF_preprint} showed that it is $2 \sqrt{\safe(G) - 2} + 1 \le \vi(G) \le 2 \safe(G)$ 
 for every connected graph $G$ (except stars) and that the bounds are tight.
 From this relation, one might speculate whether the concept of vertex integrity is essentially equivalent to that of safe number.
 This is not in fact true as their complexity differs for some cases:
 \begin{itemize}
   \item \textit{Interval graphs}:
   For unweighted interval graphs, 
   both the safe number and the vertex integrity can be determined in polynomial time~\cite{KratschKM97,AguedaCFLMMMNOS_structural}.
  However, for weighted interval graphs, determining the safe number is NP-hard as mentioned above,
   but it is still polynomial-time solvable for the vertex integrity~\cite{RayKZJ06}.
 
   \item \textit{Split graphs}:
   For unweighted split graphs, the vertex integrity is trivially determined~\cite{LiZZ08},
   while the safe number is NP-hard to determine~\cite{AguedaCFLMMMNOS_structural}.
   Determining the vertex integrity is NP-hard for weighted split graphs~\cite{DrangeDH16}.
 
   \item \textit{Parameterized complexity}:
   The vertex integrity problem parameterized by the solution size $k$ 
   admits a kernel of $k^{3}$ vertices even for the weighted version~\cite{DrangeDH16}.
   On the other hand,
   our (unweighted) problem does not admit any polynomial kernel parameterized by the solutions size 
   unless $\mathrm{PH} = \Sigma^{\mathrm{p}}_{3}$ (see Corollary~\ref{cor:no_poly-kernel}).
 \end{itemize}

 The \emph{$\ell$-component order connectivity} of a graph is the minimum integer $k$ such that
 there is a vertex set $S \subseteq V$ with $|S| \le k$ and each component of $G - S$ has size at most $\ell$.
 Such an $S$ is an \emph{$\ell$-size separator}.
 Because of the second parameter $\ell$, this parameter cannot be directly compared to the safe number.
 We summarize their similarities and differences:
 \begin{itemize}
   \item \textit{Graph classes}:
   For weighted interval graphs~\cite{Ben-AmeurMN15,DrangeDH16} and weighted circular-arc graphs~\cite{Ben-AmeurMN15} 
   the problem can be solved in polynomial time even if $\ell$ is part of the input.
 
   \item \textit{Parameterized complexity}:
   When parameterized only by $\ell$ or the solution size $k$,
   the problem is W[1]-hard even for (unweighted) split graphs~\cite{DrangeDH16}.
   There is a $2^{O(k \log \ell)} \cdot n$-time algorithm,
   which is tight in the sense that there is no $2^{o(k \log \ell)} \cdot n^{O(1)}$-time algorithm unless the ETH fails~\cite{DrangeDH16}.
   It is known that the problem admits a kernel of $O(k \ell)$ vertices~\cite{Xiao17a} 
   (see also \cite{DrangeDH16,KumarL16} for previous results on polynomial kernels).
   When parameterized by both the treewidth $\tw$ and the solution size $k$, the problem is W[1]-hard,
   and the ETH implies that there is no $f(\tw + k) \cdot n^{o((\tw+k)/ \log (\tw+k))}$-time algorithm~\cite{BonnetBKM17}.
   
   \item \textit{Approximation}:
   An $(\ell+1)$-approximation is obtained by the greedy algorithm~\cite{Ben-AmeurMN15}.
   There is a $2^{O(\ell)} n + n^{O(1)}$-time approximation algorithm of a factor $O(\log \ell)$,
   but there is no approximation algorithm that runs in time polynomial both in $n$ and $\ell$ and
   with an approximation factor $n^{(1/\log\log n)^{c}}$ unless the ETH fails~\cite{Lee18}.
 \end{itemize}
 
 There is another graph parameter, called the \emph{fracture number},
 introduced in \cite{DvorakEGKO17} to provide a way to tackle integer linear programs with small ``backdoors''.
 Given a bipartite graph $G$, its fracture number is the minimum $k$ such that there is a vertex set $S$ of size at most $k$
 while each component of $G-S$ has size at most $k$.
 Among other results, it is shown in \cite{DvorakEGKO17} that 
 the problem is NP-hard in general,
 but can be solved in time $O((k+1)^{k} \cdot m)$,
 and there is also an approximation algorithm with ratio $k+1$.

 Note that the property of having a constant vertex integrity or a constant $\ell$-component order connectivity is minor closed.
 This is not true for the safe number and the connected safe number.
 For example, let $C_{8}'$ be the graph obtained from $C_{8}$, the cycle of eight vertices,
 by adding a vertex $v$ adjacent to a diagonal pair of distance 4 in $C_{8}$.
 Since $N[v]$ is a connected safe set, $\safe(C_{6}) \le \csafe(C_{6}) \le 3$.
 On the other hand, it can be shown that $\safe(C_{8}) = \csafe(C_{8}) = 4$.


\section{W[2]-hardness parameterized by pathwidth}\label{sec_whard_pw}

In this section we show that \textsc{Safe Set} is W[2]-hard parameterized by
pathwidth, via a reduction from \textsc{Dominating Set}. 
 
Given an instance $[G=(V,E),k]$ of \textsc{Dominating Set}, we will construct an instance $[G'=(V',E'),\pw(G')]$ of
\textsc{Safe Set} parameterized by pathwidth. Let $V=\{v_1,\dots,v_n\}$ and $k'=1+kn+\sum_{v\in V}k(\delta(v)+1)$
be the target size of the safe set in the new instance, where $\delta(v)$ is the degree of $v$ in $G$.

Before proceeding, let us give a high-level description of some of the key ideas of
our reduction (an overview is given in Figure \ref{fig:dom_global}.
First, we note that our new instance will include a universal
vertex. This simplifies things, as such a vertex must be included in any safe
set (of reasonable size) and ensures that the safe set is connected. The
problem then becomes: can we select $k'-1$ additional vertices so that their
deletion disconnects the graph into components of size at most $k'$.

The main part of our construction consists of a collection of $k$ cycles of
length $n^2$. By attaching an appropriate number of leaves to every $n$-th
vertex of such a cycle we can ensure that any safe set that uses exactly $n$
vertices from each cycle must space them evenly, that is, if it selects the
$i$-th vertex of the cycle, it also selects the $(n+i)$-th vertex, the
$(2n+i)$-th vertex, etc. As a result, we expect the solution to invest $kn$
vertices in the cycles, in a way that encodes $k$ choices from $[n]$.
 
We must now check that these choices form a dominating set of the original
graph $G$. For each vertex of $G$ we construct a gadget and connect
this gadget to a different length-$n$ section of the cycles. This type of
connection ensures that the construction will in the end have small pathwidth,
as the different gadgets are only connected through the highly-structured ``choice''
part that consists of the $k$ cycles. We then construct a gadget for each
vertex $v_i$ that can be broken down into small enough components by deleting
$k(\delta(v_i)+1)$ vertices, if and only if we have already selected from the
cycles a vertex corresponding to a member of $N[v]$ in the original graph.

\paragraph{Domination gadget:} Before we go on to describe in
detail the full construction, we describe a \emph{domination gadget}
$\hat{D}_i$.  This gadget refers to the vertex $v_i\in V$ and its purpose is to
model the domination of $v_i$ in $G$ and determine the member of $N[v_i]$ that
belongs to the dominating set. We construct $\hat{D}_i$ as follows,
while Figure \ref{fig:dom_gadget} provides an illustration: 

\begin{itemize} 
\item We make a \emph{central} vertex $z^i$. We attach to this vertex
$k'-k(\delta(v_i)+1)$ leaves. Call this set of leaves $W^i$.
\item We make $k$ independent sets $X^i_1,\dots,X^i_k$ of size $|N[v_i]|$.  For
each $j\in[1,k]$ we associate each $x\in X^i_j$ with a distinct member of
$N[v_i]$. We attach to each vertex $x$ of each $X^i_j, j\in[1,k]$ an
independent set of $k'-1$ vertices. Call this independent set $Q_x$.
\item We then make another $k$ independent sets $Y^i_1,\dots,Y^i_k$ of size
$|N[v_i]|$. For each $j\in[1,k]$ we construct a perfect matching from $X^i_j$
to $Y^i_j$. We then connect $z$ to all vertices of $Y^i_j$ for all $j\in[1,k]$. 
\end{itemize}

\begin{figure}[htb]
    \centering
    \includegraphics[scale=1]{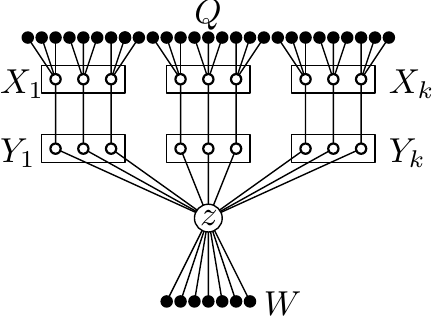}
    \caption{Our domination gadget $\hat{D}_i$. Superscripts $i$ are omitted.}
   \label{fig:dom_gadget}
  \end{figure}
The intuition behind this construction is the following: we will connect the
vertices of $X^i_j$ to the $j$-th selection cycle, linking each vertex with the
element of $N[v_i]$ it represents. In order to construct a safe set, we need to
select at least one vertex $y\in Y^i_j$ for some $j\in[1,k]$, otherwise the
component containing $z$ will be too large. This gives us the opportunity to
\emph{not} place the neighbor $x\in X^i_j$ of $y$ in the safe set.
This can only happen, however, if the neighbor of $x$ in the main part is also in the
safe set, i.e.\ if our selection dominates $v_i$.

\paragraph{Construction:} Graph $G'$ is constructed as follows:
\begin{itemize} \item We first
make $n$ copies of $V$ and serially connect them in a long cycle, conceptually
divided into $n$ \emph{blocks}:
$v_1,v_2,\dots,v_n|v_{n+1},\dots,v_{2n}|v_{2n+1},\dots|\dots,v_{n^2}|v_1$. Each
block corresponds to one vertex of $V$.
\item We make $k$ copies
$V^1,\dots,V^k$ of this cycle, where $V^i=\{v_1^i,\dots,v_{n^2}^i\},\forall
i\in[1,k]$ and refer to each as a \emph{line}. Each such line will correspond
to one vertex of a dominating set in $G$.
\item We add a set $B_j^i$ of
$k'-n+1$ neighbors to each vertex $v^i_{(j-1)n+1},\forall j\in[1,n]$, i.e.\ the
first vertex of every block of every line. We refer to these sets collectively
as the \emph{guards}. 
\item Then, for each \emph{column} of blocks (i.e.\ the
$i$-th block of all $k$ lines for $i\in[1,n]$), we make a domination gadget
$\hat{D}_i$ that refers to vertex $v_i\in V$. As described above, the gadget
contains $k$ copies $X^i_1,\dots,X^i_k$ of $N[v_i]\subseteq V$.
\item For $i \in [1,n]$ and $j \in [1,k]$, we add an edge between each vertex in $X^i_j$ and its 
corresponding vertex in the $i$-th block of $V^j$, i.e.\ for the given $i$-th column, we connect
a vertex from the $j$-th line ($V^j$) to the vertex from the $j$-th copy of
$N[v_i]$ ($X^i_j$) if they correspond to the same original vertex from $V$.
\item We add a \emph{universal vertex} $u$ and connect it to every other vertex
in the graph. This concludes our construction (see also Figure \ref{fig:dom_global}).
\end{itemize}

\begin{figure}[htb]
  \centering
  \includegraphics[scale=1]{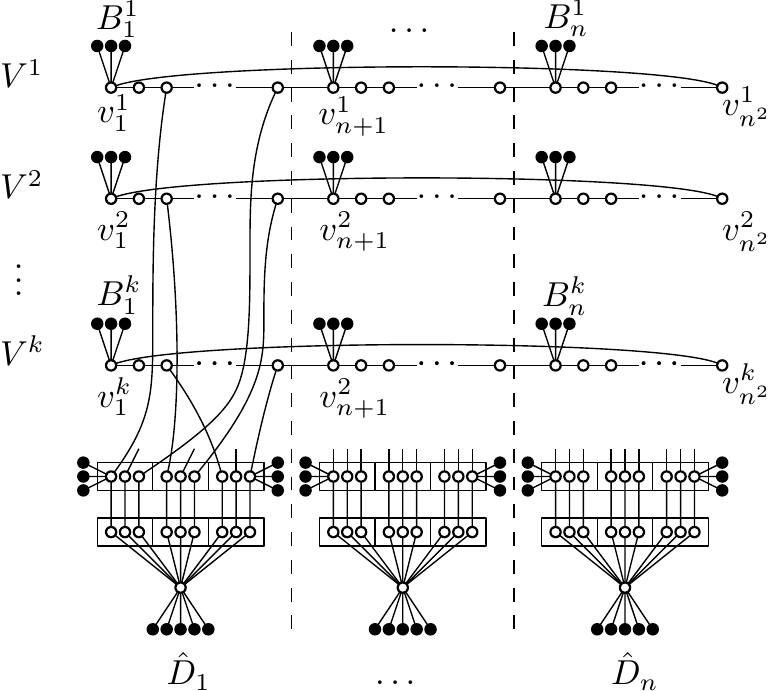}
  \caption{A simplified picture of our construction. Note sets $Q$ are only shown for two vertices per gadget, exemplary connections between corresponding vertices in sets $X$ and lines $V$ are only shown for the first gadget, while the universal vertex $u$ is omitted.}
  \label{fig:dom_global}
\end{figure}
 
 \begin{lemma}\label{whard_pw_lemFWD}
  If $G$ has a dominating set of size $k$, then $G'$ has a safe set of size $k'=1+kn+\sum_{v\in V}k(\delta(v)+1)$.
 \end{lemma}
 \begin{proof}
  Given a dominating set $K$ in $G$ of size $|K|=k$, we will construct a safe set $S$ in $G'$ of size $|S|=k'=1+kn+\sum_{v\in V}k(\delta(v)+1)$. First, we include in $S$ the universal vertex $u$. We then arbitrarily order the vertices of $K$ and will include in $S$ the $n$ copies (one from each block) of each vertex of $K$ from the line that matches this ordering, i.e.\ if the $i$-th vertex of $K$ is $v_j$, then set $S$ will include all vertices $v_j^i,v_{j+n}^i,\dots,v_{j+n^2-n}^i$ for each $i\in[1,k]$.
  
  Then, for each column and domination gadget $\hat{D}_i$, we consider the closed neighborhood $N[v_i]$ of vertex $v_i$ in $G$: as $K$ is a dominating set, $K\cap N[v_i]\not=\emptyset$, i.e.\ there must be at least one vertex $x$ from $N[v_i]$ included in $K$. If this is the $j$-th vertex in $K$ according to our ordering, we then include all vertices from sets $X_1^i,\dots,X_k^i$ in $S$, \emph{except} for the copy of $x$ in $X_j^i$ and we also include in $S$ the vertex $y$ from set $Y_j^i$ that is paired with $x$ (and is adjacent to it). Note also that set $S$ must already include a vertex corresponding to $x$ from line $V^j$. Thus $k(\delta(v_i)+1)$ vertices are selected from $\hat{D}_i$, one from the sets $Y$ and all but one from the sets $X$, while these selections are complementary. Repeating this process for all gadgets $\hat{D}_i$ (i.e.\ all columns) completes set $S$ and we claim it is a safe set for $G'$. Observe that the size of $S$ is now $|S|=k'=1+kn+\sum_{v\in V}k(\delta(v)+1)$.
  
  Since $S$ includes vertex $u$ that is connected to all other vertices, then set $S$ is connected and for it to be safe, the size of any \emph{remaining} connected component in $G' - S$ must be no larger than $k'=1+kn+\sum_{v\in V}k(\delta(v)+1)$. First, from every line $V^i$, the set $S$ periodically includes every $n$-th vertex (all corresponding to some vertex of the dominating set solution $K$) and 
 if some $v_l^i \in V^i$ with a neighbor in a dominating gadget $\hat{D}_j$ is not included in $S$, 
 then the neighbor in the gadget must be included in $S$ instead. 
 Thus the size of any remaining component in $V^i$ is at most $|B_j^i|+n-1=k'-n+1+n-1=k'=|S|$. Within a domination gadget $\hat{D}_i$, there is only one vertex $x$ from sets $X_j^i$ that does not belong in $S$, while both its neighbors from $V^j$ and $Y_j^i$ are in $S$. The size of any remaining component that includes  $x$ is equal to $|Q_x|+1=k'=|S|$. Finally, the size of any remaining component that contains the central vertex $z^i$ is equal to $1+|W^i|+k(\delta(v_i)+1)-1=1+k'-k(\delta(v_i)+1)+k(\delta(v_i)+1)-1=k'=|S|$, due to set $S$ containing only one vertex from sets $Y_j^i$ (that is adjacent to the only vertex from $X_j^i$ that is not in $S$). As there is no remaining component in $G'- S$ that is larger than $S$, our solution $S$ is safe.
 \end{proof}

 \begin{lemma}\label{whard_pw_lemBWD}
  If $G'$ has a safe set of size $k'=1+kn+\sum_{v\in V}k(\delta(v)+1)$, then $G$ has a dominating set of size $k$.
 \end{lemma}
 \begin{proof} Given a safe set $S$ of size $k'$ in $G'$, we will construct a
 dominating set $K$ of size $k$ in $G$. First, we may assume without loss of
 generality that set $S$ includes the universal vertex $u$, as it being adjacent
 to every other vertex would imply that otherwise the size of $S$ would have to
 be at least $|V'|/2$.
 
   We now make an easy observation that we may assume (without loss of
 generality) that the safe set we are given does not contain any leaves: if the
 safe set contains a leaf $u$ and does not contain its neighbor, we can exchange
 $u$ with its neighbor; if the safe set contains $u$ and its neighbor, we can
 exchange $u$ with some arbitrary non-leaf vertex, and this does not affect the
 feasibility of the solution thanks to the universal vertex.  
 
 Next, due to the guard vertices $B_j^i$ connected to every first vertex of
 every block of every line, any safe set of size $k'$ must include at least one
 vertex per block per line and these must be periodical so as to allow only
 sub-paths of length $n$ between consecutive selections: the size of every
 $B_j^i$ being $k'-n+1$, if there is such a sub-path between any pair of
 selections that is longer than $n$, then there would be a remaining component
 that is larger than $k'$. Thus set $S$ must include at least $kn+1$ vertices
 outside the domination gadgets $\hat{D}_i$.

  Next consider a gadget $\hat{D}_i$ that corresponds to vertex $v_i\in V$. The size of the set $W^i$ being $k'-k(\delta(v_i)+1)$ implies that there must be at least one vertex in $S$ from the closed neighborhood of central vertex $z^i$, as $|N[z^i]|=|W^i|+1+k(\delta(v_i)+1)=k'+1$. Furthermore, for any vertex $x$ in some $X_j^i$, set $S$ must include either vertex $x$ itself, or at least two vertices from the closed neighborhood of $x$ as the size of each set $Q_x$ is $k'-1$ and $|N[x]|=k'+2$. This means set $S$ must include at least $k(\delta(v_i)+1)$ vertices from each gadget $\hat{D}_i$ (completing the budget for $S$) and furthermore, that it is only possible for $S$ to not include some vertex $x\in X_j^i$ if both its neighbor $y\in Y_j^i$ (its pair) and its neighbor $v_l^j$ in the line $V^j$ (the pair's corresponding vertex) are included in $S$.
  
  Having thus identified the necessary structure of $S$, we construct our dominating set $K$ in $G$ by considering the selections for $S$ from every line $V^i$: as explained above, these selections must be periodical for each line (every $n$-th vertex) and for each $i\in[1,k]$ we include in $K$ the vertex from $V$ that corresponds to every copy from every block that is included in $S$. Note that these need not necessarily be distinct, i.e.\ the selections from a pair of lines might correspond to the same original vertex from $G$, meaning $K$ may be of size less than $k$, yet if $S$ is a safe set of size $k'$ in $G'$, then the resulting set $K$ must be a dominating set in $G$.
  
  Consider each column $i$ of blocks and each domination gadget $\hat{D}_i$. As explained above, the size of $S\cap \hat{D}_i$ is $k(\delta(v_i)+1)$, with at least one of these belonging in the closed neighborhood of central vertex $z^i$, while for each $x\in X_j^i$ it must be either $x\in S$, or $y\in S$ and $v_l^j\in S$, where $y$ is the pair of $x$ in $Y_j^i$ and $v_l^j$ is the pair's corresponding vertex in the line $V^j$. The budget for $\hat{D}_i$ implies there must be at least one $x$ not included in $S$, whose pair $y$ must be included instead (for the remaining component containing central vertex $z^i$ to be of size $\le k'$) and thus, the other neighbor of $x$ in line $V^j$ must also have been included in $S$, i.e.\ set $S$ must already include every copy of $v_l^j$ from every block of line $V^j$. This implies that for the vertex $v_i\in V$ in question (i.e.\ corresponding to column $i$ and gadget $\hat{D}_i$), at least one vertex from its closed neighborhood $N[v_i]$ will be included in $K$ and vertex $v_i$ will thus be dominated by $K$ in $G$. As this must hold for all columns $i$/gadgets $\hat{D}_i$ and vertices $v_i\in V$, set $K$ will be a dominating set in $G$.
 \end{proof}

 \begin{lemma}\label{whard_pw_lempwbound}
 The pathwidth $\pw(G')$ of $G'$ is at most  $2k + 4$.
 \end{lemma}
 \begin{proof} 
 
 We consider the graph $G''$ obtained from $G'$ by deleting the universal vertex
 as well as the edges connecting $v^j_1$ to $v^j_{n^2}$ for each $j\in[1,k]$. We
 will show that this graph has pathwidth at most $k+3$. This will imply the claim as
 we can obtain a path decomposition of $G'$ by adding to all bags the universal
 vertex and the vertices $v^j_1$ for all $j\in[1,k]$.
 
 Consider now the $i$-th block of the construction, consisting of the gadget $\hat{D}_i$,
 the vertices $\{v^j_{(i-1)n+1},\ldots,v^j_{in}\}$ and $B_{i}^{j}$ for $j\in[1,k]$,
 and additionally the vertex $v^j_{in+1}$, if $i < n$.
 We show that the graph induced by these vertices admits a path
 decomposition of width $k+3$ with the additional property that
 the first bag of the decomposition is $\{v^j_{(i-1)n+1} \mid 1 \le j \le k\}$,
 and the last bag of the decomposition is $\{v^j_{in+1} \mid 1 \le j \le k\}$, if $i < n$. 
 If we prove this claim we are done, since then
 we can take such a decomposition for each block and identify the last bag of
 the decomposition for the $i$-th block with the first bag of the decomposition
 for the $(i+1)$-th block to obtain a decomposition for $G''$ of the same
 width.
 
 Let us now consider the $i$-th block. We will make use of the following
 standard pathwidth fact: for any graph $H$, if $H'$ is the graph obtained by
 deleting all leaves of $H$, then $\pw(H)\le\pw(H')+1$. 
 This can easily be seen by first considering a path decomposition of $H'$ and then, for each
 leaf $u$ of $H$, finding a bag of the original decomposition that contains the
 neighbor of $u$ and inserting immediately after it a copy of the same bag with
 $u$ added. Furthermore, we can do this without changing the first and the last bags in the decomposition.
 We use this fact to bound the pathwidth of the $i$-th block as
 follows: we add the central vertex $z$ to all bags, and therefore may delete it
 from the graph. Now, removing all leaves from the block eliminates all $B$, $Y$, $Q$, and
 $W$ vertices. Removing again all leaves from this graph eliminates $X$. As a result,
 the new graph is a collection of $k$ disjoint paths, for which we can easily
 construct the required decomposition.  
\end{proof}
 
 \begin{theorem}\label{whard_pw_thm}
\textsc{Safe Set} and \textsc{Connected Safe Set} are W[2]-hard parameterized by the pathwidth of the input graph.
Furthermore, both problems cannot be solved in time $n^{o(\pw)}$ unless the ETH is false. 
\end{theorem}
 \begin{proof}
  Given an instance $[G=(V,E),k]$ of \textsc{Dominating Set}, we use the above construction to create an instance $[G'=(V',E'),\pw(G')]$ of \textsc{Safe Set} parameterized by pathwidth. Lemmas \ref{whard_pw_lemFWD} and \ref{whard_pw_lemBWD} show the correctness of our reduction, while Lemma \ref{whard_pw_lempwbound} provides the bound on the pathwidth of the constructed graph, showing that our new parameter $\pw(G')$ is linearly bounded by the original parameter $k$. The running time bound follows from the fact that, under the ETH, \textsc{Dominating Set} does not admit an $n^{o(k)}$ algorithm.

  The results for CSS follow by the connectivity of the safe set in Lemmas \ref{whard_pw_lemFWD} and \ref{whard_pw_lemBWD}.
 \end{proof}
 
 
\section{No polynomial kernel parameterized by vertex cover number}\label{sec:kernel}

A set $X \subseteq V$ is a \emph{vertex cover} of $G = (V,E)$
if each edge $e \in E$ has at least one endpoint in $X$.
The minimum size of a vertex cover in $G$ is the \emph{vertex cover number} of $G$, denoted by $\vc(G)$. 
Parameterized by vertex cover number $\vc$, both problems are FPT (see Figure~\ref{fig:width-parameters}) and
in this section we show the following kernelization hardness of SS and CSS.
\begin{theorem}
\label{thm:no_poly-kernel}
\textsc{Safe Set} and \textsc{Connected Safe Set}
parameterized by the vertex cover number
do not admit polynomial kernels even for connected graphs
unless $\mathrm{PH} = \Sigma^{\mathrm{p}}_{3}$.
\end{theorem}
Since for every graph $G$, it is $\csafe(G)/2 \le \safe(G) \le \vc(G)$ (\cite{AguedaCFLMMMNOS_structural}),
the above theorem implies that SS and CSS parameterized by the natural parameters
do not admit a polynomial kernel.
\begin{corollary}
\label{cor:no_poly-kernel}
SS and CSS 
parameterized by solution size do not admit polynomial kernels even for connected graphs
unless $\mathrm{PH} = \Sigma^{\mathrm{p}}_{3}$.
\end{corollary}

Let $P$ and $Q$ be parameterized problems.
A polynomial-time computable function $f \colon \Sigma^{*} \times N \to \Sigma^{*} \times N$
is a \emph{polynomial parameter transformation} from $P$ to $Q$
if there is a polynomial $p$ such that for all $(x,k) \in \Sigma^{*} \times N$, it is:
$(x,k) \in P$ if and only if $(x',k') = f(x,k) \in Q$, and
$k' \le p(k)$.
If such a function exits, then $P$ is \emph{polynomial-parameter reducible} to $Q$.

\begin{proposition}
[\cite{BodlaenderDFH09}]
\label{prop:ppt-poly-kernel}
Let $P$ and $Q$ be parameterized problems,
and $P'$ and $Q'$ be unparameterized versions of $P$ and $Q$, respectively.
Suppose $P'$ is NP-hard, $Q'$ is in NP,
and $P$ is polynomial-parameter reducible to $Q$.
If $Q$ has a polynomial kernel, then $P$ also has a polynomial kernel.
\end{proposition}

To prove Theorem~\ref{thm:no_poly-kernel},
we present a polynomial-parameter transformation
from the well-known \textsc{Red-Blue Dominating Set} problem (RDBS) to SS (and CSS) parameterized by vertex cover number.
RDBS becomes trivial when $k \ge |R|$ and thus we assume that $k < |R|$ in what follows.
It is known that RBDS parameterized simultaneously by $k$ and $|R|$
does not admit a polynomial kernel unless $\mathrm{PH} = \Sigma^{\mathrm{p}}_{3}$~\cite{DomLS14}.
Since RBDS is NP-hard and SS and CSS are in NP,
it suffices to present a polynomial-parameter transformation
from RBDS parameterized by $k + |R|$ to SS and CSS parameterized by the vertex cover number.

From an instance $[G, k]$ of RBDS with $G = (R,B;E)$, 
we construct an instance $[H, s]$ of SS as follows
 (see \figref{fig:no_poly-kernel}).
Let $s = k + |R| + 1$.
We add a vertex $u$ to $G$ and make it adjacent to all vertices in $B$.
We then attach $2s$ pendant vertices to each vertex in $R$ and to $u$.
Finally, for each $r \in R$, we make a star $K_{1,s-1}$ and add an edge between $r$ and the center of the star.
We call the resultant graph $H$.
Observe that $\vc(H) \le 2|R|+1$
since $\{u\} \cup R \cup C$ is a vertex cover of $H$,
where $C$ is the set of centers of stars attached to $R$.
This reduction is a polynomial-parameter transformation
from RBDS parameterized by $k + |R|$
to SS parameterized by  vertex cover number.

If $D \subseteq V(G)$ is a solution of RBDS of size $k$, 
then $S := \{u\} \cup R \cup D$ is a connected safe set of size $s$.
To see this, recall that $B$ is an independent set and $N_{H}(B) = \{u\} \cup R$.
Thus each component in $H - S$ is either an isolated vertex in $D \setminus B$,
or a star $K_{1,s-1}$ with $s$ vertices.

Assume that $(H,s)$ is a yes instance of SS and let $S$ be a safe set of $H$ with $|S| \le s$.
Observe that $\{u\} \cup R \subseteq S$ since $u$ and all vertices in $R$ have degree at least $2s$.
Since $|S \setminus (\{u\} \cup R)| \le k < |R|$,
$S$ cannot intersect all stars attached to the vertices in $R$.
Hence $H - S$ has a component of size at least $|V(K_{1,s-1})| = s$.
This implies that $S$ is connected.
Since $R$ is an independent set and 
each path from $u$ to a vertex in $R$ passes through $B$,
each vertex in $R$ has to have a neighbor in $B \cap S$.
Thus $B \cap S$ dominates $R$.
Since $B \cap S \subseteq S \setminus (\{u\} \cup R)$, it has size at most $k$.
Therefore, $[G, k]$ is a Yes instance of RBDS.\@
This completes the proof of Theorem~\ref{thm:no_poly-kernel}.

 \begin{figure}[bth]
   \centering
   \includegraphics[scale=1]{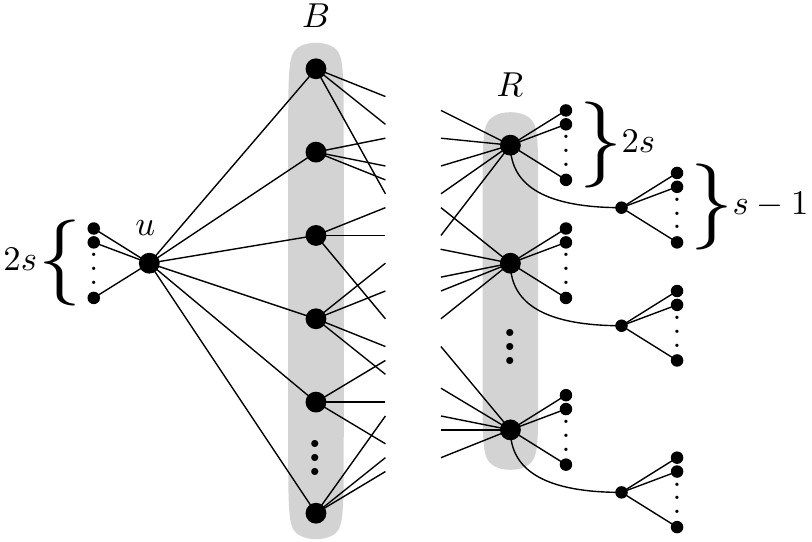}
   \caption{The graph $H$ in our reduction from RBDS to SS for the proof of Theorem~\ref{thm:no_poly-kernel}.}
   \label{fig:no_poly-kernel}
 \end{figure}
 
 
\section{FPT algorithm parameterized by neighborhood diversity}\label{sec_fpt_nd}

In this section, we present FPT algorithms for SS and CSS parameterized by neighborhood diversity.
That is, we prove the following theorem.
\begin{theorem}
\label{thm:nd-fpt}
\textsc{Safe Set} and \textsc{Connected Safe Set} are fixed-parameter tractable when parameterized by the neighborhood diversity.
\end{theorem}

In a graph $G = (V,E)$, two vertices $u,v \in V$ are \emph{twins} if $N(u) \setminus \{v\} = N(v) \setminus \{u\}$.
The \emph{neighborhood diversity} $\nd(G)$ of $G = (V,E)$ is the minimum integer $k$ such that
$V$ can be partitioned into $k$ sets $T_{1}, \dots, T_{k}$ of pairwise twin vertices.
It is known that such a minimum partition can be found in linear time using fast modular decomposition algorithms~\cite{McConnellS99,TedderCHP08}.
It is also known that $\nd(G) \le 2^{\vc(G)} + \vc(G)$ for every graph $G$~\cite{Lampis12}.

Let $G$ be a connected graph such that $\nd(G) = k$.
Let $T_{1}, \dots, T_{k}$ be the partition of $V(G)$ into sets of pairwise twin vertices.
Note that each $T_{i}$ is either a clique or an independent set by definition.
We assume that $k \ge 2$ since otherwise the problem becomes trivial.
Since each $T_{i}$ is a twin set, the sizes of intersections $|S \cap T_{i}|$ completely characterize
the sizes of the components in $G[S]$ and $G-S$, and the adjacency among them.

Let $S \subseteq V(G)$ and $s_{i} = |S \cap T_{i}|$ for $i \in [k]$.
We partition $[k]$ into $I_{\mathbf{f}}$, $I_{\mathbf{p}}$, and $I_{\emptyset}$ as follows:
\begin{align}
  i \in 
  \begin{cases}
    I_{\emptyset}  & \text{if $s_{i} = 0$,} \\
    I_{\mathbf{p}} & \text{if $1 \le s_{i} \le |T_{i}|-1$,}  \\
    I_{\mathbf{f}} & \text{otherwise ($s_{i} = |T_{i}|$).} 
  \end{cases} \label{eq:nd-index-partition}
\end{align}
For $i, i' \notin I_{\emptyset}$ (not necessarily distinct),
twin sets $T_{i}$ and $T_{i'}$ are \emph{reachable} in $S$
if either 
\begin{itemize}
  \item $T_{i} = T_{i'}$ and $T_{i}$ is a clique, or

  \item there is a sequence $i_{0}, \dots, i_{\ell}$ of indices
  such that $\ell \ge 1$, $i_{0} = i$, $i_{\ell} = i'$, $i_{j} \notin I_{\emptyset}$ for all $j$,
  and $T_{i_{j}}$ and $T_{i_{j+1}}$ are adjacent for $0 \le j < \ell$.
\end{itemize}
 \begin{lemma}
 \label{lem:nd-single}
 If $i \notin I_{\emptyset}$ and $T_{i}$ is not reachable to $T_{i}$ itself in $S$,
 then each vertex in $T_{i} \cap S$ induces a component of size $1$.
 \end{lemma}
 \begin{proof}
 The assumption implies that $T_{i}$ is an independent set
 and $i' \in I_{\emptyset}$ for each $T_{i'}$ adjacent to $T_{i}$.
 Therefore, each vertex in $T_{i} \cap S$ has no other vertex in $S$
 that belongs to the same component of $G[S]$.
  
 \end{proof}
 
 \begin{lemma}
 \label{lem:nd-nonsingle}
 Two vertices $u, v \in S$ are in the same component of $G[S]$
 if and only if 
 $u \in T_{i}$ and $v \in T_{i'}$ for some $i$ and $i'$, and $T_{i}$ is reachable from $T_{i'}$ in $S$.
 \end{lemma}
 \begin{proof}
 If $u,v \in T_{i}$ for some $i$ where $T_{i}$ is a clique,
 then $\{u,v\} \in E(G)$ and thus they are in the same component of $G[S]$.
 Assume that $u \in T_{i}$ and $v \in T_{i'}$, and
 there is a sequence $i_{0}, \dots, i_{\ell}$ with $\ell \ge 1$ that shows the reachability.
 We show that $u,v$ are in the same component of $G[S]$ by induction on $\ell$.
 If $\ell = 1$, then $T_{i}$ and $T_{i'}$ are adjacent and we are done.
 Let $\ell > 1$. Since $i_{\ell-1} \notin I_{\emptyset}$, there is a vertex $w \in T_{\ell-1} \cap S$.
 By the induction hypothesis, $u$ and $w$ are in the same component of $G[S]$.
 Furthermore, since $T_{\ell-1}$ and $T_{\ell} = T_{i'}$ are adjacent,
 there is an edge between $w$ and $v$. Hence $u$ and $v$ are in the same component of $G[S]$.
 
 Now assume that $u \in T_{i}$ and $v \in T_{i'}$ are in the same connected connected of $G[S]$.
 Then there is a shortest $u$--$v$ path $(p_{0}, \dots, p_{q})$ in $G[S]$.
 If $T_{i} = T_{i'}$ and $T_{i}$ is a clique, then $p_{0}, p_{q} \in T_{i}$.
 Assume that this is not the case. Let $T_{i_{j}}$ be the twin set including $p_{j}$.
 Since the path is shortest, $i_{j} \ne i_{j+1}$ holds
 and $T_{j}$ is adjacent to $T_{j+1}$ for each $j < q$.
 Also, because of $p_{j}$, we have $i_{j} \notin I_{\emptyset}$ for each $j$.
 Therefore, $T_{i}$ and $T_{i'}$ are reachable in $S$.
\end{proof}

By Lemmas~\ref{lem:nd-single} and \ref{lem:nd-nonsingle},
each component of $G[S]$ is either a single vertex,
or the intersection of $S$ and the union of a maximal family of pairwise reachable twin sets in $S$.
Observe that the maximal families of pairwise reachable twin sets in $S$
is determined only by the set $I_{\emptyset}$.
Also, if $R$ is a maximal family of pairwise reachable twin sets in $S$,
then the corresponding component of $G[S]$ has size $\sum_{T_{i} \in R} |T_{i} \cap S|$.

Now, just by interchanging the roles of $S$ and $V(G) \setminus S$ in Lemmas~\ref{lem:nd-single} and \ref{lem:nd-nonsingle},
we can show the following counterparts that 
imply that the maximal families of pairwise reachable twin sets in $V(G) \setminus S$
are determined only by the set $I_{\mathbf{f}}$,
while the size of the component of $G-S$ corresponding to a maximal family $R$ of pairwise reachable twin sets in $V(G) \setminus S$
is $\sum_{T_{i} \in R} |T_{i} \setminus S|$.
\begin{lemma}
\label{lem:nd-single-co}
If $i \notin I_{\mathbf{f}}$ and $T_{i}$ is not reachable to $T_{i}$ itself in $V(G) \setminus S$,
then each vertex in $T_{i} \setminus S$ induces a component of size $1$.
\end{lemma}
\begin{lemma}
\label{lem:nd-nonsingle-co}
Two vertices $u, v \in V(G) \setminus S$ are in the same component of $G - S$
if and only if 
$u \in T_{i}$ and $v \in T_{i'}$ for some $i$ and $i'$, and $T_{i}$ is reachable from $T_{i'}$ in $V(G) \setminus S$.
\end{lemma}

\subsection*{ILP formulation}

Now we reduce the problem to an FPT number of integer linear programs with a bounded number of variables.
We first divide $[k]$ into the subsets $I_{\emptyset}$, $I_{\mathbf{p}}$, $I_{\mathbf{f}}$ in Eq.~\eqref{eq:nd-index-partition}.
There are $3^{k}$ candidates for such a partition.

For each $i \in [k]$, we use a variable $x_{i}$ to represent the size of $T_{i} \cap S$.
To find a minimum safe set satisfying $I_{\emptyset}$, $I_{\mathbf{p}}$, $I_{\mathbf{f}}$,
we set the objective function to be $\sum_{i \in [k]} x_{i}$
and minimize it subject to the following linear constraints.
The first set of constraints is to make $S$ consistent with the guess of $I_{\emptyset}$, $I_{\mathbf{p}}$, $I_{\mathbf{f}}$:
\begin{align*}
  x_{i} = 0 & \quad \text{for $i \in I_{\emptyset}$},\\
  1 \le x_{i} \le |T_{i}|-1 & \quad \text{for $i \in I_{\mathbf{p}}$},\\
  x_{i} = |T_{i}| & \quad \text{for $i \in I_{\mathbf{f}}$}.
\end{align*}

As discussed above, the set of sizes $x_{i}$ completely characterizes
the structure of components in $G[S]$ and $G - S$.
In particular, we can decide whether $G[S]$ is connected or not at this point.
We reject the disconnected case if we are looking for a connected safe set.

Let $\mathcal{C}$ and $\mathcal{D}$
be the sets of maximal families of pairwise reachable twin sets in $S$ and $V(G) \setminus S$, respectively.
Note that the twin sets that satisfy the conditions of Lemma~\ref{lem:nd-single} (Lemma~\ref{lem:nd-single-co})
are not included in any member of $\mathcal{C}$ ($\mathcal{D}$, respectively).

For each $C_{j} \in \mathcal{C}$, we use a variable $y_{j}$ to represent the size of the corresponding component of $G[S]$.
Also, for each $D_{h} \in \mathcal{D}$, we use a variable $z_{h}$ to represent the size of the corresponding component of $G - S$.
This can be stated as follows:
\begin{align*}
  y_{j} = \sum_{T_{i} \in C_{j}} x_{i}  & \quad \text{for $C_{j} \in \mathcal{C}$}, \\
  z_{h} = \sum_{T_{i} \in D_{h}} |T_{i}| - x_{i}  & \quad \text{for $D_{h} \in \mathcal{D}$}. 
\end{align*}

We say that $C_{j} \in \mathcal{C}$ and $D_{h} \in \mathcal{D}$ are \emph{touching}
if there are $T \in C_{j}$ and $T' \in D_{h}$ that are adjacent, or the same.
We can see that $C_{j}$ and $D_{h}$ are touching if and only if the corresponding components are adjacent
via an edge from $T$ and $T'$ or an edge completely in $T = T'$. 
We add the following constraint to guarantee the safeness of $S$:
 \begin{align*}
   y_{j} \ge z_{h} & \quad \text{for each pair of touching $C_{j} \in \mathcal{C}$ and $D_{h} \in \mathcal{D}$}. 
 \end{align*}

Now we have to deal with the singleton components of $G[S]$ (we can ignore the singleton components of $G - S$
because the components of $G[S]$ adjacent to them have size at least 1).
Let $T_{i}$ be a twin set that satisfies the conditions of Lemma~\ref{lem:nd-single}.
That is, $T_{i} \cap S \ne \emptyset$, $T_{i}$ is an independent set,
and no twin set adjacent to $T_{i}$ has a non-empty intersection with $S$.
Hence a component of $G - S$ is adjacent to the singleton components in $T_{i} \cap S$,
if and only if the corresponding family $D_{h} \in \mathcal{D}$ includes a twin set adjacent to $T_{i}$.
We say that such $D_{h}$ is \emph{adjacent to} $T_{i}$.
Therefore, we add the following constraint:
 \begin{align*}
   z_{h} \le 1 & \quad
   \text{for each $D_{h} \in \mathcal{D}$ adjacent to $T_{i}$ satisfying Lemma~\ref{lem:nd-single}.}
 \end{align*}

\subsection*{Solving the ILP}

Lenstra~\cite{Lenstra83} showed that the feasibility of an ILP formula can be decided in FPT time
when parameterized by the number of variables (see also~\cite{Kannan87,FrankT87}).
Fellows et al.~\cite{FellowsLMRS08} extended it to the optimization version.
More precisely, we define the problem as follows:
\begin{myproblem}
  \problemtitle{\textsc{$p$-Opt-ILP}}
  \probleminput{A matrix $A \in \mathbb{Z}^{m \times p}$, and vectors $b \in \mathbb{Z}^{m}$ and $c \in \mathbb{Z}^{p}$.}
  \problemquestion{Find a vector $x \in \mathbb{Z}^{p}$ that minimizes $c^{\top} x$ and satisfies that $A x \ge b$.}
\end{myproblem}
\noindent
They then showed the following:
\begin{theorem}
[Fellows et al.~\cite{FellowsLMRS08}]
\label{thm:ilp-opt}
\textsc{$p$-Opt-ILP} can be solved using $O(p^{2.5 p + o(p)} \cdot L \cdot \log(MN))$ arithmetic operations
and space polynomial in $L$, where $L$ is the number of bits in the input,
$N$ is the maximum absolute values any variable can take,
and $M$ is an upper bound on the absolute value of the minimum taken by the objective function.
\end{theorem}

In the formulation for SS and CSS, we have at most $O(k)$ variables:
$x_{i}$ for $i \in [k]$,
$y_{j}$ for $C_{j} \in \mathcal{C}$, and
$z_{h}$ for $D_{h} \in \mathcal{D}$.
Observe that the elements of $\mathcal{C}$ (and of $\mathcal{D}$ as well) are pairwise disjoint.
We have only $O(k^{2})$ constraints
and the variables and coefficients can have values at most $|V(G)|$.
Therefore, Theorem~\ref{thm:nd-fpt} holds.
 

\section{XP algorithm parameterized by clique-width}\label{sec_xp_cw}

This section presents an XP-time algorithm for SS and CSS parameterized by clique-width.
The algorithm runs in time $O(g(c) \cdot n^{f(c)})$, where $c$ is the clique-width. It is known that for any constant $c$,
  one can compute a $(2^{c+1}-1)$-expression of a graph of clique-width $c$ in $O(n^{3})$ time~\cite{HlinenyO08,OumS06,Oum08}.
 We omit $g(c)$ in the running time and focus on the exponent $f(c)$.
 
%
  
 %
 
 \begin{theorem}
\label{thm:cliquewidth-xp}
Given an $n$-vertex graph $G$ and an irredundant $c$-expression $T$ of $G$,
the values of $\safe(G)$ and $\csafe(G)$,
along with their corresponding sets can be computed in $O(n^{28 \cdot 2^{c} + 1})$ time.
\end{theorem}

\begin{corollary}
\label{cor:cliquewidth-xp}
Given an $n$-vertex graph $G$,
the values of $\safe(G)$ and $\csafe(G)$,
along with their corresponding sets can be computed in time 
$n^{O(f(\cw(G)))}$, where $f(c) = 2^{2^{c+1}}$.
\end{corollary}

 For each node $t$ in the $c$-expression $T$ of $G$, 
 let $G_{t}$ be the vertex-labeled graph represented by $t$.
 We denote by $V_{t}$ the vertex set of $G_{t}$.
 For each $i$, we denote the set of $i$-vertices in $G_{t}$ by $V_{t}^{i}$.
 For sets $S \subseteq V_{t}$ and $L \subseteq [c]$,
 we denote by $\mathcal{C}(S,L)$ and $\mathcal{D}(S,L)$ the set of components of $G_{t}[S]$ and $G_{t}-S$, respectively,
 that include exactly the labels in $L$.
 
 For each node $t$ in $T$, we construct a table 
 $\DP_{t}(\num, \bnum, \smllst, \lrgst, \asmllst, \alrgst, \mindiff) \in \{\true, \false\}$
 with indices 
 $\num, \bnum \colon 2^{[c]} \to \{0,\dots,n\}$,
 $\smllst \colon 2^{[c]} \to \{0,\dots,n\} \cup \{\infty\}$,
 $\lrgst \colon 2^{[c]} \to \{0,\dots,n\} \cup \{-\infty\}$,
 $\asmllst \colon 2^{[c]} \times 2^{[c]} \to \{0,\dots,n\} \cup \{\infty\}$,
 $\alrgst \colon 2^{[c]} \times 2^{[c]} \to \{0,\dots,n\} \cup \{-\infty\}$, and
 $\mindiff \colon 2^{[c]} \times 2^{[c]} \to \{-n,\dots,n\} \cup \{\infty\}$.
 We set $\DP_{t}(\num, \bnum, \smllst, \lrgst, \asmllst, \alrgst, \mindiff) =  \true$
 if and only if there exists a set $S \subseteq V_{t}$ such that, for all $L \subseteq [c]$:\footnote{%
 We assume that $\min$ and $\max$ return $\infty$ and $-\infty$, respectively, when applied to the empty set.}
 \begin{itemize}
   \item $\num = \sum_{C \in \mathcal{C}(S,L)} |C|$, \ $\bnum = \sum_{D \in \mathcal{D}(S,L)} |D|$,
   \item $\smllst = \min_{C \in \mathcal{C}(S,L)} |C|$, \ $\lrgst  = \max_{D \in \mathcal{D}(S,L)} |D|$,
 \end{itemize} 
 and for all $L_{1}, L_{2} \subseteq [c]$:
 \begin{itemize}
   \item $\asmllst = \min \{|C| \mid C \in \mathcal{C}(S,L_{1}), D \in \mathcal{D}(S, L_{2}), C \text{ and } D \text{ are adjacent}\}$,
   \item $\alrgst = \max \{|D| \mid C \in \mathcal{C}(S,L_{1}), D \in \mathcal{D}(S, L_{2}), C \text{ and } D \text{ are adjacent}\}$,
   \item $\mindiff = \min \{|C| - |D| \mid C \in \mathcal{C}(S,L_{1}), D \in \mathcal{D}(S, L_{2}), C \text{ and } D \text{ are adjacent}\}$.
 \end{itemize}
 
 Let $r$ be the root node of $T$.
 Observe that $\safe(G)$ is the minimum integer $s$ such that
 there exist $\num$, $\bnum$, $\smllst$, $\lrgst$, $\asmllst$, $\alrgst$, and $\mindiff$ satisfying that
 $\DP_{r}(\num, \bnum, \smllst, \lrgst, \asmllst, \alrgst, \mindiff) =  \true$,
 $s = \sum_{L \subseteq [c]} \num$, and $\mindiff(L_{1}, L_{2}) \ge 0$ for all $L_{1}, L_{2} \subseteq [c]$.
 For $\csafe(G)$, we additionally ask that
 $\num(L)$ is nonzero with exactly one $L \subseteq [c]$.
 
 We can compute in a bottom-up manner all entries
 $\DP_{t}(\num, \bnum, \smllst, \lrgst, \asmllst, \alrgst, \mindiff)$.
 Note that there are $O(n^{10 \cdot 2^{c} + 1})$ such entries.
 Hence, to prove Theorem~\ref{thm:cliquewidth-xp}, it suffices to show that
 each entry can be computed in time $O(n^{18\cdot 2^{c}})$
 assuming that the entries for the children of $t$ are already computed. This is indeed the case for a $\cup$-node, while for $\rho$- and $\eta$-nodes $O(n^{10\cdot 2^c})$ will suffice.

 \begin{lemma}
 \label{lem:cw-leaf}
 For a leaf node $t$ with label $\circ_{i}$,
 $\DP_{t}(\num, \bnum, \smllst, \lrgst, \asmllst, \alrgst, \mindiff)$
 can be computed in $O(1)$ time.
 \end{lemma}
 \begin{proof}
 Observe that $\DP_{t}(\num, \bnum, \smllst, \lrgst, \asmllst, \alrgst, \mindiff) = \true$ if and only if:
 \begin{itemize}
   \item $\num(L) = \bnum(L) = 0$, $\smllst(L) = \infty$, and $\lrgst(L) = -\infty$ for all $L \ne \{i\}$,
 
   \item $\asmllst(L_{1}, L_{2}) = \infty$, $\alrgst(L_{1}, L_{2}) = -\infty$, and $\mindiff(L_{1}, L_{2}) = \infty$
   for all $L_{1}, L_{2} \subseteq [c]$,
 \end{itemize}
 and either:
 \begin{itemize}
   \setlength{\itemsep}{0pt}
   \item $\num(\{i\}) = 0$, $\bnum(\{i\}) = 1$, $\smllst(\{i\}) = \infty$, $\lrgst(\{i\}) = 1$,
   or
   \item $\num(\{i\}) = 1$, $\bnum(\{i\}) = 0$, $\smllst(\{i\}) = 1$, $\lrgst(\{i\}) = -\infty$,
 \end{itemize}
 where the first case corresponds to $S = \emptyset$ and the second one to $S = V_{t}^{i} = V_{t}$.
 These conditions can be checked in $O(1)$ time.
 \end{proof}
 
 \begin{lemma}
 For a $\cup$-node $t$, 
 $\DP_{t}(\num, \bnum, \smllst, \lrgst, \asmllst, \alrgst, \mindiff)$
 can be computed in $O(n^{18 \cdot 2^{c}})$ time.
 \end{lemma}
 \begin{proof}
 Let $t_{1}$ and $t_{2}$ be the children of $t$ in $T$.
 Now, $\DP_{t}(\num, \bnum, \smllst, \lrgst, \asmllst, \alrgst, \mindiff) = \true$ if and only if there exist
 tuples $(\num_{1}, \bnum_{1}, \smllst_{1}, \lrgst_{1}, \asmllst_{1}, \alrgst_{1}, \mindiff_{1})$ and
 $(\num_{2}, \bnum_{2}, \smllst_{2}, \lrgst_{2}, \asmllst_{2}, \alrgst_{2}, \mindiff_{2})$
 such that:
 \begin{itemize}
   \item $\DP_{t_{1}}(\num_{1}, \bnum_{1}, \smllst_{1}, \lrgst_{1}, \asmllst_{1}, \alrgst_{1}, \mindiff_{1}) = \true$,
   \item $\DP_{t_{2}}(\num_{2}, \bnum_{2}, \smllst_{2}, \lrgst_{2}, \asmllst_{2}, \alrgst_{2}, \mindiff_{2}) = \true$,
   \item $\num(L) = \num_{1}(L) + \num_{2}(L)$ for all $L \subseteq [c]$,
   \item $\bnum(L) = \bnum_{1}(L) + \bnum_{2}(L)$ for all $L \subseteq [c]$,
   \item $\smllst(L) = \min\{\smllst_{1}(L), \smllst_{2}(L)\}$ for all $L \subseteq [c]$,
  \item $\lrgst(L) = \max\{\lrgst_{1}(L), \lrgst_{2}(L)\}$ for all $L \subseteq [c]$,
  \item $\asmllst(L_{1},L_{2}) = \min\{\asmllst_{1}(L_{1},L_{2}), \asmllst_{2}(L_{1},L_{2})\}$ for all $L_{1}, L_{2} \subseteq [c]$,
   \item $\alrgst(L_{1},L_{2}) = \max\{\alrgst_{1}(L_{1},L_{2}), \alrgst_{2}(L_{1},L_{2})\}$ for all $L_{1},L_{2} \subseteq [c]$, and
   \item $\mindiff(L_{1},L_{2}) = \min\{\mindiff_{1}(L_{1},L_{2}), \mindiff_{2}(L_{1},L_{2})\}$ for all $L_{1}, L_{2} \subseteq [c]$.
 \end{itemize}
 
 There are at most $n^{2^{c}}$ possible pairs for $(\num_{1}, \num_{2})$ as $\num_{2}$ is uniquely determined by $\num_{1}$.
 Similarly, there are at most $n^{2^{c}}$ possible pairs for $(\bnum_{1}, \bnum_{2})$.
 There are at most $n^{2 \cdot 2^{c}}$ candidates each for $(\smllst_{1}, \smllst_{2})$ and $(\lrgst_{1}, \lrgst_{2})$, and
 at most $n^{4 \cdot 2^{c}}$ candidates each for $(\asmllst_{1}, \asmllst_{2})$, $(\alrgst_{1}, \alrgst_{2})$,
 and $(\mindiff_{1}, \mindiff_{2})$.
 In total, there are at most $n^{18 \cdot 2^{c}}$ candidates for the tuples
 $(\num_{1}, \bnum_{1}, \smllst_{1}, \lrgst_{1}, \asmllst_{1}, \alrgst_{1}, \mindiff_{1})$ and
 $(\num_{2}, \bnum_{2}, \smllst_{2}, \lrgst_{2}, \asmllst_{2}, \alrgst_{2}, \mindiff_{2})$.
 Each candidate can be checked in $O(1)$ time, and thus the lemma holds.
\end{proof}
 
 
 \begin{lemma}
 For a $\rho_{i,j}$-node $t$, 
 $\DP_{t}(\num, \bnum, \smllst, \lrgst, \asmllst, \alrgst, \mindiff)$
 can be computed in $O(n^{10 \cdot 2^{c}})$ time.
 \end{lemma}
 \begin{proof}
 Let $t_{1}$ be the child of $t$ in $T$.
 Observe that a component with label set $L$ in $G_{t}$
 has label set either $L$, $L^{+i}$, or $L^{+i-j}$ in $G_{t_{1}}$,
 where $L^{+i} = L \cup \{i\}$ and $L^{+i-j} = L \cup \{i\} \setminus \{j\}$.
 Thus $\DP_{t}(\num, \bnum, \smllst, \lrgst, \asmllst, \alrgst, \mindiff) = \true$
 if and only if there exists
 a tuple $(\num_{1}, \bnum_{1}, \smllst_{1}, \lrgst_{1}, \asmllst_{1}, \alrgst_{1}, \mindiff_{1})$ such that
 $\DP_{t_{1}}(\num_{1}, \bnum_{1}, \smllst_{1}, \lrgst_{1}, \asmllst_{1}, \alrgst_{1}, \mindiff_{1}) = \true$,
 where for all $L \subseteq [c]$:
 \begin{itemize}
   \item if $i \in L$, then
   $\num(L) = \bnum(L) = 0$, $\smllst(L) = \infty$, and $\lrgst(L) = -\infty$;
 
   \item if $i, j \notin L$, then
   $\num(L) = \num_{1}(L)$, $\bnum(L) = \bnum_{1}(L)$, $\smllst(L) = \smllst_{1}(L)$, and $\lrgst(L) = \lrgst_{1}(L)$;
 
   \item if $i \notin L$ and $j \in L$, then:
   \begin{itemize}
     \item $\num(L) = \num_{1}(L) + \num_{1}(L^{+i}) + \num_{1}(L^{+i-j})$,
     \item $\bnum(L) = \bnum_{1}(L) + \bnum_{1}(L^{+i}) + \bnum_{1}(L^{+i-j})$,
     \item $\smllst(L) = \min\{\smllst_{1}(L), \smllst_{1}(L^{+i}), \smllst_{1}(L^{+i-j})\}$, and
     \item $\lrgst(L) = \max\{\lrgst_{1}(L), \lrgst_{1}(L^{+i}), \lrgst_{1}(L^{+i-j})\}$;
   \end{itemize}
 \end{itemize}
 and for all $L_{1}, L_{2} \subseteq [c]$:
 \begin{itemize}
   \item if $i \in L_{1} \cup L_{2}$, then $\asmllst(L_{1}, L_{2}) = \infty$, $\alrgst(L_{1},L_{2}) = -\infty$, and $\mindiff(L_{1},L_{2}) = \infty$;
 
   \item if $i,j \notin L_{1} \cup L_{2}$, then
   $\asmllst(L_{1}, L_{2}) = \asmllst_{1}(L_{1}, L_{2})$, 
   $\alrgst(L_{1},L_{2}) = \alrgst_{1}(L_{1},L_{2})$, and
   $\mindiff(L_{1},L_{2}) = \mindiff_{1}(L_{1},L_{2})$;
 
   \item if $i \notin L_{1} \cup L_{2}$ and $j \in L_{1} \cup L_{2}$, then
   \begin{align*}
     \asmllst(L_{1}, L_{2})
     &= \min\{\asmllst(L_{1}', L_{2}') \mid L_{1}' \in \mathcal{L}_{1}, L_{2}' \in \mathcal{L}_{2}\},
     \\
     \alrgst(L_{1}, L_{2})
     &= \max\{\alrgst(L_{1}', L_{2}')  \mid L_{1}' \in \mathcal{L}_{1}, L_{2}' \in \mathcal{L}_{2}\},
     \\
     \mindiff(L_{1}, L_{2})
     &= \min\{\mindiff(L_{1}', L_{2}') \mid L_{1}' \in \mathcal{L}_{1}, L_{2}' \in \mathcal{L}_{2}\},
   \end{align*}
   where, for $h \in \{1,2\}$,
   $\mathcal{L}_{h} = \{L_{h}\}$ if $j \notin L_{h}$, and
   $\mathcal{L}_{h} = \{L_{h}, L_{h}^{+i}, L_{h}^{+i-j}\}$ if $j \in L_{h}$.
 \end{itemize}
 The claimed running time follows from the fact that there are $O(n^{10 \cdot 2^{c}})$ candidates for 
 $(\num_{1}, \bnum_{1}, \smllst_{1}, \lrgst_{1}, \asmllst_{1}, \alrgst_{1}, \mindiff_{1})$
 and that each candidate can be checked in $O(1)$ time.
\end{proof}
 
 
 \begin{lemma}
 For an $\eta_{i,j}$-node $t$,
 $\DP_{t}(\num, \bnum, \smllst, \lrgst, \asmllst, \alrgst, \mindiff)$
 can be computed in $O(n^{10 \cdot 2^{c}})$ time.
 \end{lemma}
 \begin{proof}
 Let $t_{1}$ be the child of $t$ in $T$.
 If there is a set $S \subseteq V(G_{t})$ corresponding to the tuple $\tau = (\num, \bnum, \smllst, \lrgst, \asmllst, \alrgst, \mindiff)$,
 then $V_{t}^{h} \cap S$ is non-empty if and only if $\num(\tau,h) := \sum_{C \in \mathcal{C}(S,L), \ h \in L \subseteq [c]} |C| > 0$.
 Also, $V_{t}^{h} - S$ is non-empty if and only if $\bnum(\tau,h) := \sum_{D \in \mathcal{D}(S,L), \ h \in L \subseteq [c]} |D| > 0$.
 Given a tuple, these conditions can be checked in linear time. Hence we assume that we know which cases apply to the tuple.

 Let $\mathcal{L}_{i,j} = \{L \subseteq [c] \mid L \cap \{i,j\} \ne \emptyset, \num(L) > 0\}$
 and $\overline{\mathcal{L}}_{i,j} = \{L \subseteq [c] \mid L \cap \{i,j\} \ne \emptyset, \bnum(L) > 0\}$.
 By slightly abusing the notation, we denote 
 by $\bigcup\mathcal{L}_{i,j}$ the union $\bigcup_{L \in \mathcal{L}_{i,j}} L$, and 
 by $\bigcup\overline{\mathcal{L}}_{i,j}$ the union $\bigcup_{L \in \overline{\mathcal{L}}_{i,j}} L$.
 If $V_{t}^{i} \cap S \ne \emptyset$ and $V_{t}^{j} \cap S \ne \emptyset$,
 then all components in $G_{t_{1}}[S]$ containing an $i$-vertex or a $j$-vertex
 will be merged into one component with the color set $\bigcup\mathcal{L}_{i,j}$ in $G_{t}[S]$.
 Otherwise, at most one color class $i$ or $j$ contains vertices of $S$.
 Hence no components of $G_{t_{1}}[S]$ will be merged in $G_{t}[S]$,
 while the edges added between $V_{t}^{i}$ and $V_{t}^{j}$ make each component in $G_{t}[S]$ with $i$-vertices ($j$-vertices),
 if any, adjacent to each component in $G_{t}-S$ with $j$-vertices ($i$-vertices, respectively).
 The analogous observations hold also for the components in $G_{t_{1}} - S$ and $G_{t} - S$.
 
 The following claims are almost direct consequences of the discussion above.
 \begin{myclaim}
   \label{clm:nonemptyLcomp}
   If $\num(L) > 0$, then one of the following cases hold:
   \begin{itemize}
     \item $i,j \notin L$;
     \item $|\{i,j\} \cap L| = 1$ and  $\num(\tau, i) \cdot \num(\tau, j) = 0$; or
     \item $i, j \in L = \bigcup\mathcal{L}_{i,j}$ and  $\num(\tau, i) \cdot \num(\tau, j) > 0$.
   \end{itemize}
 \end{myclaim}
 \begin{proof}
 If none of the cases above holds, then we have either
 \begin{enumerate}
   \item $\{i,j\} \cap L \ne \emptyset$, $L \ne \bigcup\mathcal{L}_{i,j}$, and $\num(\tau, i) \cdot \num(\tau, j) > 0$; or
 
   \item $i, j \in L$ and $\num(\tau, i) \cdot \num(\tau, j) = 0$.
 \end{enumerate}
 In the first case, $\num(\tau, i) \cdot \num(\tau, j) > 0$ implies that 
 there is a unique component with color class $L$ such that $L \cap \{i,j\} \ne \emptyset$.
 Furthermore, we know that $L = \bigcup\mathcal{L}_{i,j}$.
 In the second case, $\num(\tau, i) = 0$ or $\num(\tau, j) = 0$ holds,
 and thus no component can contain both $i$-vertices and $j$-vertices.
\end{proof}
 We can prove the next claim in the same way.
 \begin{myclaim}
   \label{clm:nonemptyLcomp-bar}
   If $\bnum(L) > 0$, then one of the following cases hold:
   \begin{itemize}
     \item $i,j \notin L$;
     \item $|\{i,j\} \cap L| = 1$ and $\bnum(\tau, i) \cdot \bnum(\tau, j) = 0$; or
     \item $i,j \in L = \bigcup\overline{\mathcal{L}}_{i,j}$ and $\bnum(\tau, i) \cdot \bnum(\tau, j) > 0$.
   \end{itemize}
 \end{myclaim}
 
 In the following, we only explain the case where each color set satisfies the condition in
 Claim~\ref{clm:nonemptyLcomp} or Claim~\ref{clm:nonemptyLcomp-bar} (depending on which function we are talking about),
 since otherwise the values of functions can be trivially determined.

 Now we can see that 
 $\DP_{t}(\num, \bnum, \smllst, \lrgst, \asmllst, \alrgst, \mindiff) = \true$
 if and only if there exists
 a tuple $(\num_{1}, \bnum_{1}, \smllst_{1}, \lrgst_{1}, \asmllst_{1}, \alrgst_{1}, \mindiff_{1})$ such that
 $\DP_{t_{1}}(\num_{1}, \bnum_{1}, \smllst_{1}, \lrgst_{1}, \asmllst_{1}, \alrgst_{1}, \mindiff_{1}) = \true$,
 where for all $L \subseteq [c]$, the following holds:
 \begin{itemize}
   \item if $i,j \notin L$, then $\num(L) = \num_{1}(L)$, $\bnum(L) = \bnum_{1}(L)$, 
   $\smllst(L) = \smllst_{1}(L)$, and $\lrgst(L) = \lrgst_{1}(L)$;
 
   \item if $|\{i,j\} \cap L| = 1$ and $\num(\tau, i) \cdot \num(\tau, j) = 0$, then $\num(L) = \num_{1}(L)$
   and $\smllst(L) = \smllst_{1}(L)$;
 
   \item if $|\{i,j\} \cap L| = 1$ and $\bnum(\tau, i) \cdot \bnum(\tau, j) = 0$, then $\bnum(L) = \bnum_{1}(L)$
   and $\lrgst(L) = \lrgst_{1}(L)$;
 
   \item if $i,j \in L = \bigcup\mathcal{L}_{i,j}$ and $\num(\tau, i) \cdot \num(\tau, j) > 0$, then 
   $\num(L) = \sum_{Q \in \mathcal{L}_{i,j}} \num_{1}(Q)$ and $\smllst(L) = \num(L)$;
   
   \item if $i,j \in L = \bigcup\overline{\mathcal{L}}_{i,j}$ and $\bnum(\tau, i) \cdot \bnum(\tau, j) > 0$, then
   $\bnum(L) = \sum_{Q \in \overline{\mathcal{L}}_{i,j}} \bnum_{1}(Q)$ and $\lrgst(L) = \bnum(L)$;
 \end{itemize}
 and for all $L_{1}, L_{2} \subseteq [c]$, the following holds:
 \begin{itemize}
 
   \item \textbf{(no related merge / no new edge)}
   if either 
   \begin{itemize}
     \item $i,j \notin L_{1} \cup L_{2}$
     
     \item $i,j \notin L_{1}$, $|\{i,j\} \cap L_{2}| = 1$, and $\bnum(\tau, i) \cdot \bnum(\tau, j) = 0$, or
 
     \item $i,j \notin L_{2}$, $|\{i,j\} \cap L_{1}| = 1$, and $\num(\tau, i) \cdot \num(\tau, j) = 0$,
   \end{itemize}
   then $\asmllst(L_{1}, L_{2}) = \asmllst_{1}(L_{1}, L_{2})$,
   $\alrgst(L_{1}, L_{2}) = \alrgst_{1}(L_{1}, L_{2})$, and $\mindiff(L_{1}, L_{2}) = \mindiff_{1}(L_{1}, L_{2})$;
 
   \smallskip
 
   \item \textbf{(related but unmerged color sets)}
   if $|\{i,j\} \cap L_{1}| = |\{i,j\} \cap L_{2}| = 1$, $\num(\tau, i) \cdot \num(\tau, j) = 0$,
   $\bnum(\tau, i) \cdot \bnum(\tau, j) = 0$, then
   \begin{itemize}
     \item if $\{i,j\} \cap L_{1} = \{i,j\} \cap L_{2}$, then
     $\asmllst(L_{1}, L_{2}) = \asmllst_{1}(L_{1}, L_{2})$,
     $\alrgst(L_{1}, L_{2}) = \alrgst_{1}(L_{1}, L_{2})$, and $\mindiff(L_{1}, L_{2}) = \mindiff_{1}(L_{1}, L_{2})$;
 
     \item if $\{i,j\} \cap L_{1} \ne \{i,j\} \cap L_{2}$, then
     $\asmllst(L_{1}, L_{2}) = \smllst(L_{1})$,
     $\alrgst(L_{1}, L_{2}) = \lrgst(L_{2})$, and 
     $\mindiff(L_{1}, L_{2}) = \smllst(L_{1}) - \lrgst(L_{2})$;
   \end{itemize}
 
   \smallskip
 
   \item \textbf{(merged color set $L_{1}$)}
   if $i,j \in L_{1} = \bigcup\mathcal{L}_{i,j}$ and $\num(\tau, i) \cdot \num(\tau, j) > 0$, then
   \begin{itemize}
     \item if $i,j \notin L_{2}$, then
     \begin{align*}
       \asmllst(L_{1}, L_{2}) &=
       \begin{cases}
 	\infty & \text{if } \min_{L_{1}' \in \mathcal{L}_{i,j}} \asmllst_{1}(L_{1}', L_{2}) = \infty, \\
 	\num(L_{1}) & \text{otherwise},
       \end{cases}
     \end{align*}
     $\alrgst(L_{1}, L_{2})  = \max_{L_{1}' \in \mathcal{L}_{i,j}} \alrgst_{1}(L_{1}', L_{2})$, and
     $\mindiff(L_{1}, L_{2}) = \asmllst(L_{1}, L_{2}) - \alrgst(L_{1}, L_{2})$;
     
     \item if either
     \begin{itemize}
       \item $|\{i,j\} \cap L_{2}| = 1$, $\bnum(\tau, i) \cdot \bnum(\tau, j) = 0$, and $\bnum(L_{2}) \ne 0$, or
       \item $i,j \in L_{2} = \bigcup\overline{\mathcal{L}}_{i,j}$ and $\bnum(\tau, i) \cdot \bnum(\tau, j) > 0$,
     \end{itemize}
     then
     $\asmllst(L_{1}, L_{2}) = \num(L_{1})$,
     $\alrgst(L_{1}, L_{2})  = \lrgst(L_{2})$, and
     $\mindiff(L_{1}, L_{2}) = \asmllst(L_{1}, L_{2}) - \alrgst(L_{1}, L_{2})$;
   \end{itemize}
 
   \smallskip
 
   \item \textbf{(merged color set $L_{2}$)}
   if $i,j \in L_{2} = \bigcup\overline{\mathcal{L}}_{i,j}$ and $\bnum(\tau, i) \cdot \bnum(\tau, j) > 0$, then
   \begin{itemize}
     \item if $i,j \notin L_{1}$, then
     $\asmllst(L_{1}, L_{2})  = \min_{L_{2}' \in \overline{\mathcal{L}}_{i,j}} \asmllst_{1}(L_{1}, L_{2}')$,
     \begin{align*}
       \alrgst(L_{1}, L_{2}) &=
       \begin{cases}
 	-\infty & \text{if } \max_{L_{2}' \in \overline{\mathcal{L}}_{i,j}} \alrgst_{1}(L_{1}, L_{2}') = -\infty, \\
 	\bnum(L_{2}) & \text{otherwise},
       \end{cases}
     \end{align*}
     $\mindiff(L_{1}, L_{2}) = \asmllst(L_{1}, L_{2}) - \alrgst(L_{1}, L_{2})$;
     
     \item if $|\{i,j\} \cap L_{1}| = 1$, $\num(\tau, i) \cdot \num(\tau, j) = 0$, and $\num(L_{1}) \ne 0$,
     then
     $\asmllst(L_{1}, L_{2}) = \smllst(L_{1})$,
     $\alrgst(L_{1}, L_{2})  = \bnum(L_{2})$, and
     $\mindiff(L_{1}, L_{2}) = \asmllst(L_{1}, L_{2}) - \alrgst(L_{1}, L_{2})$;
   \end{itemize}
 \end{itemize}
 
 The number of candidate tuples $(\num_{1}, \bnum_{1}, \smllst_{1}, \lrgst_{1}, \asmllst_{1}, \alrgst_{1}, \mindiff_{1})$
 is $O(n^{10 \cdot 2^{c}})$ and each of them can be checked in $O(1)$ time.
\end{proof}


\section{Faster algorithms parameterized by solution size}\label{sec_fpt_k}

We know that both SS and CSS admit FPT algorithms~\cite{AguedaCFLMMMNOS_structural} when parameterized by the solution size.
The algorithms in \cite{AguedaCFLMMMNOS_structural} use Courcelle's theorem~\cite{Courcelle92}, however, 
and thus their dependency on the parameter may be gigantic. The natural question would be whether they admit $O^{*}(k^{k})$-time algorithms as is the case for vertex integrity~\cite{DrangeDH16}.

We answer this question with the following theorems.
The first step of our algorithm for SS is a branching procedure to first guess the correct number of components ($k$ choices) and then guess their sizes (at most $k^k$ choices). We complete our solutions (ensuring they are connected) by constructing and solving appropriate \textsc{Steiner Tree} sub-instances. With a simple modification our algorithm also works for CSS.

\begin{theorem}\label{thm_fpt_k}
  \textsc{Safe Set} can be solved in $O^{*}(2^{k}k^{3k})$ time, where $k$ is the solution size.
\end{theorem}
 \begin{proof}
 
 As before, we assume that the input graph is connected, otherwise we solve the
 problem on each component and take the minimum. Suppose that the input graph
 $G=(V,E)$ has a safe set $O\subseteq V$ of size $k$. Since $G[O]$ is not
 necessarily connected, suppose that $G[O]$ has $\ell$ components, for $\ell\le
 k$, $O_1,\ldots, O_\ell$. The first step of our algorithm is to guess the value of
 $\ell$ ($k$ choices) and then guess the $\ell$ sizes of the components of
 $G[O]$ (at most $k^\ell\le k^k$ choices).  For all $i \in [\ell]$ we
 denote $k_i:=|O_i|$. In the remainder we assume that the algorithm guessed
 these values correctly, as we will repeat it for all possible values.
 
 The first phase of our algorithm is a branching process. We maintain $\ell$
 sets of vertices $S_1,\ldots,S_\ell$ with the intuitive meaning that
 $S_i\subseteq O_i$ for all $i\in [\ell]$. We denote
 $S:=\cup_{i=1}^{\ell} S_i$. Initially $S:=\emptyset$. If at any point we have
 $|S_i|>k_i$ for some $i$, the algorithm rejects this branch.
 
 We will say that a vertex $u\in V\setminus S$ is \emph{problematic} if it
 fulfills one of the following two properties: (i) the component of $G - S$
  that contains $u$ has size at least $k+1$, or (ii) there exists an
 $i\in [\ell]$ such that $N(u)\cap S_i\neq \emptyset$ and the component
 of $G - S$ that contains $u$ has size at least $k_i+1$. It is not
 hard to see from this description that finding a problematic vertex can be done
 in polynomial time. 
 
 The main branching step of our algorithm is now the following: we check if
 there exists a problematic vertex $u$. If it does, then we find a set of
 vertices $C$ such that $u\in C$, $C$ is connected in $G - S$, and $C$
 fulfills the following properties: (i) if the component of $u$ in $G - S$
  has size at least $k+1$, then $|C|=k+1$ or, (ii) if the component of $u$ in
 $G - S$ has size at least $k_i+1$ for some $i\in [\ell]$
 with $N(u)\cap S_i\neq \emptyset$, then $|C|=k_i+1$. It is not hard to see that
 this can be done in polynomial time. Now we produce $|C|\ell \le (k+1)k$ branches:
 for each vertex $w\in |C|$ we consider the case where we place $w$ in $S_i$ for
 $i\in [\ell]$. We then call the same algorithm recursively.
 
 To see the correctness of this branching procedure, we observe that if the
 branching is correct thus far, that is, if $S_i\subseteq O_i$ and $k_i=|O_i|$
 for all $i$, then $O\cap C\neq \emptyset$. In case (i) this is clear as we have
 a connected set of size $k+1$. In case (ii) this is a consequence of the fact
 that $C$ is adjacent to $S_i$ (therefore, to $O_i$), but has size strictly
 larger than the size we guessed (correctly) for $O_i$.
 
 The above search-tree process produces $O(k^{2k})$ branches, since we have at most $(k+1)k$
 choices in each branch, and the branching depth is at most $k$ (since in
 each branch we add a vertex to $S$).  The branching procedure terminates either
 because $|S_i|>k_i$ for some $i$ (in which case we reject), or because there
 are no more problematic vertices in the graph. Let us now explain how to
 complete the solution in this case.
 
 At this point we have $\ell$ sets $S_1,\ldots,S_\ell$ with the property that
 for all $i\in [\ell]$, every component of $G - S$ that is
 adjacent to $S_i$ has size at most $k_i$. However, we do not have a feasible
 solution yet because $G[S_i]$ is not necessarily connected for all $i$. We
 therefore construct $\ell$ instances of \textsc{Steiner Tree}: for each
 $i\in [\ell]$, we construct an instance where the set of terminals
 that must be connected is $S_i$ and we solve \textsc{Steiner Tree} on the graph
 $G - (\cup_{j\neq i} S_j)$ (in other words, for each set $S_i$ we
 find a Steiner Tree that connects $S_i$ without using any of the vertices of
 $S\setminus S_i$).  We execute the algorithm of \cite{Nederlof13}, which runs
 in time $O^{*}(2^t)$, where $t$ is the number of terminals. The algorithm
 returns $\ell$ sets $S_1',S_2',\ldots,S_\ell'$, such that for all
 $i\in [\ell]$ we have $S_i\subseteq S_i'$ and $G[S_i']$ is connected.
 If for some $i$ we have $|S_i'|>k_i$ or there is no solution (because deleting
 $S\setminus S_i$ disconnects $S_i$) we reject this branch.  Otherwise, for each
 $i$ such that $|S_i'|<k$, we augment $S_i'$ by adding to it arbitrary neighbors
 so that it remains connected and we have in the end $|S_i'|=k_i$. We return
 $S':=\cup_{i=1}^\ell S_i'$ as our solution.

 To see that the above algorithm is correct, we first argue that if we return a
 solution $S'$, then $S'$ clearly has size $k$, so we only need to explain why
 $S'$ is a safe set. Suppose for contradiction that $S'$ is not a safe set,
 therefore there exists a component of $G[S']$ which is adjacent to a component
 $C$ of $G - S'$ of larger size. In other words, there exists an
 $i\in [\ell]$ such that $G[S_i']$ has a neighbor in a component $C$
 of $G - S'$ with $|C|>|S_i'|=k_i$.  Let $S_i$ be the set of terminals
 on which we ran the $i$-th \textsc{Steiner Tree} procedure, and $S$ the set of
 vertices we had when the branching procedure stopped.  We will show that the
 graph contained a problematic vertex (with respect to $S_i$), and therefore the
 branching procedure could not have stopped.  Take a vertex $u_1\in C$ that has
 a neighbor in $S_i'$ and consider a shortest path in $G - S$ from
 $u_1$ to a vertex that has a neighbor in $S_i$.  Such a path exists, because in
 $G - (S\setminus S_i)$ all vertices of $S_i$ are in the same
 component (otherwise there would be no Steiner Tree connecting them), which is
 the same component that contains all vertices of $S_i'$, and $u_1$ has a
 neighbor in $S_i'$. Let $u_2$ be the last vertex of this path from $u_1$ to
 $S_1$. Then $u_2$ is problematic:  we have $N(u_2)\cap S_i\neq \emptyset$, and
 the component that $u_2$ belongs to in $G - S$ is at least as large
 as $C$.
 
 For the other direction, if $O$ exists and we have guessed the $k_i$ values
 correctly and $S_i\subseteq O_i$ for all $i$, then the \textsc{Steiner Tree}
 instances will all return a solution that we accept, as $O_i$ itself is a valid
 solution of the proper size. Therefore, if we reject in this phase it implies
 that no solution exists with the guessed properties.
 
 Finally, for the running time, we repeat the algorithm $O^*(k^k)$ times (for
 each value of $\ell$, and values of $k_i$), each repetition has a branching
 step with $O^*(k^{2k})$ leaves, and in each leaf we run $\ell$ times the
 \textsc{Steiner Tree} algorithm, each with at most $k$ terminals, therefore
 taking at most $O^*(2^k)$. Thus, the total running time is at most
 $O^*(2^k k^{3k})$.  
 \end{proof}
  
 When we set $\ell = 1$, the algorithm above will find a connected safe set of size at most $k$ (if one exists).
 In that case, we have a single execution of the branching algorithm with the search tree size $O(k^{k})$
 in which we execute the $O^*(2^k)$-time \textsc{Steiner Tree} algorithm for each leaf.
 Thus we have the following corollary.
\begin{corollary}
  \textsc{Connected Safe Set} can be solved in $O^{*}(2^{k} k^{k})$ time, where $k$ is the size of the solution.
\end{corollary}
 
 \newpage

\bibliography{safe}

\end{document}